\documentclass[runningheads]{llncs}
\usepackage{graphicx}
\let\lncsproof\proof \let\lncsendproof\endproof \let\lncsqed\qed
\let\proof\relax\let\endproof\relax
\usepackage{amsmath,amssymb,amsthm}
\let\proof\lncsproof \let\endproof\lncsendproof \let\qed\lncsqed
\usepackage{subcaption}
\usepackage{graphicx}
\usepackage{stmaryrd}
\usepackage{cite}
\usepackage{xcolor}
\usepackage{adjustbox}
\usepackage{prftree}
\usepackage{version}
\usepackage{ulem}
\usepackage{listings}
\usepackage{multirow}
\usepackage{makecell}
\usepackage{cases}
\usepackage{url}
\usepackage{hyperref}
\hypersetup{hidelinks,
	colorlinks=true,
	allcolors=black,
	pdfstartview=Fit,
	breaklinks=true}

\usepackage[firstpage]{draftwatermark}  

\includeversion{conf}
\excludeversion{rep}

\newcommand\ode[2]{\langle #1 \mathop{\&} #2\rangle}
\newcommand\hoare[3]{\{#1\}#2\{#3\}}
\newcommand\inv[3]{\llbracket #1 \rrbracket \langle #2 \rangle \llbracket #3\rrbracket}

\newcommand\vc{\mathrm{vc}}
\newcommand\pre{\mathrm{pre}}
\newcommand\ichoice{\mathop{\texttt{++}}}
\newcommand\skipcmd{\mathtt{skip}}
\newcommand\itecmd[3]{\mathtt{if}\ #1\allowbreak\ \mathtt{then}\ #2\allowbreak\ \mathtt{else}\ #3}
\newcommand\loopcmd[1]{#1*}
\newcommand\dL{$\mathsf{d}\mathcal{L}$}

\newcommand\vctag[1]{\textsf{(#1)}}

\newcommand\labeq{$_=$}

\newcommand\labneq{$_{\neq}$}
\newcommand\labineq{$_{\succcurlyeq}$}

\lstdefinelanguage{HHLPy}{%
  keywords={},
  sensitive=true,
  basicstyle=\scriptsize,
  morecomment=[l]{\#},
  commentstyle=\color{vgreen},
  morekeywords = [2]{pre,post,if,else,else if,function,invariant},
  keywordstyle = [2]{\bf\color{black}},
  morekeywords = [3]{solution, di, dbx, bc},
  keywordstyle = [3]{\color{purple}},
  numbers=left,
  numberstyle=\tiny,
  breaklines=true,
  frame=single
}

\newtheorem*{thm-vcg-sound}{Theorem~\ref{thm:vcg-sound}}

\begin{document}
%
%
\title{\textsf{HHLPy}: Practical Verification of Hybrid Systems using Hoare Logic
}
%
%
\author{Huanhuan Sheng\inst{1,2}
\and
Alexander Bentkamp\inst{1}
\and
Bohua Zhan\inst{1,2}
}
\authorrunning{H. Sheng, A. Bentkamp et al.}
%
\institute{State Key Laboratory of Computer Science, Institute of Software, Chinese Academy of Sciences, Beijing, China\\
\email{\{shenghh,bentkamp,bzhan\}@ios.ac.cn}
\and
University of Chinese Academy of Sciences, Beijing, China
}
\maketitle              
\SetWatermarkAngle{0}
\SetWatermarkText{\raisebox{12.5cm}{%
   \hspace{0cm}%
   \href{https://doi.org/10.5281/zenodo.7419103}{\includegraphics{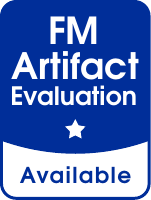}}%
   \hspace{11cm}%
   \includegraphics{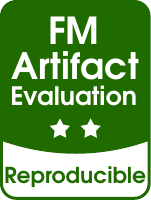}%
}}
\begin{abstract}
We present a tool for verification of hybrid systems expressed in the sequential fragment of HCSP (Hybrid Communicating Sequential Processes). The tool permits annotating HCSP programs with pre- and postconditions, invariants, and proof rules for reasoning about ordinary differential equations. Verification conditions are generated from the annotations following the rules of a Hoare logic for hybrid systems. We designed labeling and highlighting mechanisms to distinguish and visualize different verification conditions. The tool is implemented in Python and has a web-based user interface. We evaluated the effectiveness of the tool on translations of Simulink/Stateflow models and on KeYmaera X benchmarks.

\keywords{Hybrid systems \and Hoare logic \and Formal verification.}
\end{abstract}
%
%
%
\section{Introduction}

Hybrid systems refer to systems that have both continuous and discrete behaviors. They occur in diverse areas of science and engineering, ranging from transportation and spaceflight, to robots and medical devices. Hence, verifying that hybrid systems meet certain specifications is an important problem. Apart from methods such as monitoring and model checking, theorem proving is one of the major approaches to verifying hybrid systems.

There is a substantial amount of previous work on verification of hybrid systems based on theorem proving. One major framework is Platzer's differential dynamic logic (\dL) \cite{Platzer08,Platzer18}, and the associated KeYmaera/KeYmaera X prover~\cite{PlatzerQ08, Fulton-2015-KeYmaera}. Recently, a Hoare logic has been introduced for \dL\ and implemented within the Isabelle proof assistant~\cite{MuniveS18}. We review these works in detail in Section~\ref{sec:related-work} of this paper.

Another approach is to model hybrid systems using HCSP (Hybrid CSP) \cite{He94, ChaochenJR95}, an extension of CSP (Communicating Sequential Processes) to include continuous evolution. Its semantics of continuous evolution is deterministic, so it can be used naturally for capturing Simulink/Stateflow models.
A hybrid Hoare logic has been developed for HCSP, and is implemented in Isabelle~\cite{WangZZ15}. However, practical application of the tool is complicated by its steep learning curve.
To use this tool, the user need to be familiar with the Isabelle proof assistant, as well as manually applying a set of Hoare logic rules which are themselves very complex.
This is in stark contrast to KeYmaera X, which allows users to verify hybrid programs by choosing menu actions and offers highly specialized automation.

This paper introduces \textsf{HHLPy}\footnote{The tool is available at \url{https://github.com/bzhan/mars/tree/master/hhlpy}}, a tool for verification of the sequential part of HCSP with a friendly graphical user interface.
Compared to~\cite{WangZZ15}, we simplify the Hoare logic rules, and add more rules for reasoning about the behavior of differential equations. 
These latter rules are closely related to that in \dL, but due to the semantic differences of HCSP, we adapted some of the rules and proved our rules to be sound (Section~\ref{sec:proof-rules}).

Our Hoare logic rules are in a sufficiently simple form that automatic verification condition generation is possible. We design such a procedure to compute verification conditions (VCs) from a given annotated HCSP program (Section~\ref{sec:verif}). We express VCs as a set of conditions, splitting up VCs that are conjunctions as much as possible. We use labels to distinguish between different VCs, so that users can choose solvers (currently either Z3 \cite{MouraB08} or Wolfram Engine \cite{WolframEngine}) for each VC individually and such choices are maintained through minor changes on the code (Section~\ref{sec:labels}).

To visualize to the user where each VC originates from,
a highlighting mechanism highlights the set of code fragments in the annotated program that contributed to generating the VC (Section~\ref{sec:highlighting}).

We implemented the tool using Python and JavaScript
and evaluated it on Simulink/Stateflow models and on KeYmaera X benchmarks (Section~\ref{sec:evaluation}). 
We translated two Simulink/Stateflow models using
the toolchain developed by Zou et al.\cite{Zou-2013-verify, Zou-2015-formal} and verified them in our tool.
Due to differences in the semantics of \dL\ and HCSP, we translated each KeYmaera X benchmark by hand, trying to maintain semantic equivalence as much as possible. 
In this way, we succeeded to use our tool to solve most of the verification problems in the basic and nonlinear KeYmaera X benchmarks.



\vspace{-1mm}
\section{Preliminaries}
\label{sec:preliminaries}

In this section, we present the sequential fragment of HCSP, with an informal explanation of its semantics. We further give an overview of the existing toolchain on translation of Simulink/Stateflow models into HCSP.

\vspace{-1mm}
\subsection{Sequential fragment of HCSP}

Hybrid CSP (HCSP), introduced in~\cite{He94,ChaochenJR95}, is an extension of Hoare's Communicating Sequential Processes to include continuous evolution. It can model communicating processes running in parallel, where each process may have both continuous and discrete behavior. In this paper, we focus on the sequential fragment of HCSP, consisting of the following commands:
\[
\begin{array}{lll}
 S,T  & ::= & 
 \skipcmd \mid
 x := e \mid
 x := *\>(B) \mid
 S; T \mid 
 \itecmd{B}{S}{T} \mid
 S \ichoice T \mid
 \loopcmd{S}\\
 &&\mid 
 \ode{\dot{\boldsymbol{x}}=\boldsymbol{e}}{D}
 \end{array}
\]
 The program state is a mapping from variables to reals. $\skipcmd$ leaves the state unchanged. $x := e$ assigns the value of expression $e$ to variable $x$. $x := *\>(B)$ is nondeterministic assignment of some value satisfying condition $B$ to $x$. $S; T$ and $\itecmd{B}{S}{T}$ are regular sequential composition and conditional. $S \ichoice T$ is a nondeterministic choice between $S$ and $T$. $\loopcmd{S}$ runs $S$ a nondeterministic number of times (including zero).

The ordinary differential equation (ODE) command $\ode{\dot{\boldsymbol{x}}=\boldsymbol{e}}{D}$
specifies continuous evolution in HCSP. It makes the vector of variables $\boldsymbol{x}$ evolve according to ODE $\dot{\boldsymbol{x}}=\boldsymbol{e}$ until the domain $D$ becomes false. If $D$ is false from the start, the ODE is skipped. In contrast to \dL, where continuous evolution may stop at any point within the specified domain, in HCSP it always deterministically continues up to the boundary. In this paper, we assume $D$ is given by a polynomial inequality of the form $p(\boldsymbol{x})<0$, so it represents an open set in $\mathbb{R}^n$.

We assume in this paper that all expressions appearing in an HCSP program (as well as in annotations to be discussed later) are polynomials, and hence continuity conditions are trivially satisfied.

For a formal treatment of semantics of HCSP (including communication and parallel composition), we refer to Zhan et al.~\cite[Chapter 6]{zhan-2017-formal}.

\subsection{Translation from Simulink/Stateflow}

The HCSP language is located at the center of a toolchain that also includes translation from Simulink/Stateflow models, simulation and code generation~\cite{ChenHTWYZZZ17}. The original translation algorithms from Simulink~\cite{Zou-2013-verify} and Stateflow~\cite{Zou-2015-formal} produce HCSP programs that involve communication between parallel processes. However, more recent methods by Xu et al.~\cite{SimulinkTranslate} and Guo et al.~\cite{StateflowTranslate} produce sequential HCSP programs. We use these translation methods for verification of Simulink/Stateflow models in Section~\ref{sec:evaluation}.

\section{Proof Rules of Hoare Logic for Hybrid Systems}
\label{sec:proof-rules}

In this section, we present the Hoare logic that forms the basis of our verification tool. The Hoare triple for partial correctness, written as $\hoare{P}{c}{Q}$, means starting from a state satisfying assertion $P$, any terminating execution of $c$ reaches a state satisfying assertion $Q$. The Hoare rules for ordinary commands (except ODEs) are standard and are presented in the Appendix \ref{sec:more-proof-rules}.

Hence, we focus on the Hoare rules for ODEs. These rules are mostly adapted from rules for \dL, as given in~\cite{Platzer18,PlatzerT20}. Due to the difference in semantics between HCSP and \dL, several of the rules take on different forms. We do not aim to present a minimal set of rules, instead providing users a wide range of choices.

\subsection{Proof rules based on invariants}

In order to state proof rules based on invariants of ODEs, we require an additional kind of judgments, called \emph{invariant triples}.

\begin{definition}[Invariant Triple]
Let $P$ and $Q$ be predicates on the
variables of an ODE $\dot{\boldsymbol{x}}=\boldsymbol{e}$.
Let
$\boldsymbol{\gamma}:[0,T]\to \mathbb{R}^n$
be a solution of the ODE
such that $\boldsymbol{\gamma}(t)$ satisfies $P$ for all $t\in[0,T]$ and such that $\boldsymbol{\gamma}(0)$ satisfies $Q$.
If for all such solutions $\boldsymbol{\gamma}$,
$\boldsymbol{\gamma}(t)$ satisfies $Q$ for all $t\in[0,T]$, then we say that $Q$ is an invariant of ODE $\dot{\boldsymbol{x}}=\boldsymbol{e}$ under domain $P$, written as
    \[\inv{P}{\dot{\boldsymbol{x}}=\boldsymbol{e}}{Q}\]
\end{definition}

\subsubsection{Differential Weakening}

The differential weakening rule (dW) reduces a Hoare triple goal to an invariant triple, incorporating the domain condition.
\[
\prftree[r]{dW}{\inv{\overline{D}}{\dot{\boldsymbol{x}}=\boldsymbol{e}}{I}}{\partial D\wedge I\rightarrow Q}
{\hoare{(D\rightarrow I)\wedge (\neg D\rightarrow Q)}{\ode{\dot{\boldsymbol{x}}=\boldsymbol{e}}{D}}{Q}}
\]
Here, $\overline{D}$ is the closure of $D$, and $\partial D$ is the boundary set of $D$. Note that the rule is in the form that allows us to derive a precondition from any postcondition. The precondition $(D\rightarrow I)\wedge (\neg D\rightarrow Q)$
corresponds to the two cases for the state before ODE: if the state satisfies domain $D$, then it should satisfy the invariant. Otherwise it should satisfy the postcondition $Q$ directly. 
Two special cases of the rule, for $I$ set to $\mathsf{true}$ and $\mathsf{false}$, provide further intuition. They correspond to cases where no invariant is needed, and where the starting state is known to satisfy $\neg D$.
\[
\begin{array}{c}
\prftree[r]{dWT}{\partial D\rightarrow Q}
{\hoare{\neg D\rightarrow Q}{\ode{\dot{\boldsymbol{x}}=\boldsymbol{e}}{D}}{Q}}
\quad
\prftree[r]{dWF}{\hoare{\neg D\wedge Q}{\ode{\dot{\boldsymbol{x}}=\boldsymbol{e}}{D}}{Q}}
\end{array}
\]

\begin{proof}[of the (dW) rule]
Given starting state $\boldsymbol{x}$, we divide into two cases based on whether $\boldsymbol{x}$ satisfies domain $D$.
If $\boldsymbol{x}$ satisfies $D$, then there exists a solution $\boldsymbol{\gamma}:[0,T]\to\mathbb{R}^n$, such that $\boldsymbol{\gamma}(t)$ satisfies $D$ for $t\in[0,T)$ and $\boldsymbol{\gamma}(T)$ satisfies $\neg D$, and we wish to show that $\boldsymbol{\gamma}(T)$ satisfies $Q$. By the continuity of $\boldsymbol{\gamma}$, we get that $\boldsymbol{\gamma}(t)$ satisfies $\overline{D}$ for $t\in[0,T]$. Moreover, since $D\rightarrow I$ holds in the precondition, we get that $\boldsymbol{\gamma}(0)$ satisfies $I$ as well. Then from $\inv{\overline{D}}{\dot{\boldsymbol{x}}=\boldsymbol{e}}{I}$, we get that $\boldsymbol{\gamma}(t)$ satisfies $I$ for $t\in [0,T]$. From $\partial D\wedge I\rightarrow Q$ and the fact that $\boldsymbol{\gamma}(T)$ satisfies $I$ and $\partial D$, we get that $\boldsymbol{\gamma}(T)$ satisfies $Q$, as desired.

If $\boldsymbol{x}$ does not satisfy $D$, then the ODE is not executed, and we wish to show that $\boldsymbol{x}$ satisfies $Q$. Since $\neg D\rightarrow Q$ holds in the precondition, we get that $\boldsymbol{x}$ satisfies $Q$, as desired.
\end{proof}


\subsubsection{Differential Invariant}

The differential invariant rule (dI) is essentially the same as that in \dL. It concludes invariants from computation of Lie derivatives.
\[
\prftree[r]{dI\labeq}{P\rightarrow \dot{f}=0}{\inv{P}{\dot{\boldsymbol{x}}=\boldsymbol{e}}{f=0}}
\]
Here $\dot{f}$ denotes the Lie derivative of $f$ under the differential equation $\dot{\boldsymbol{x}}=\boldsymbol{e}$.
The corresponding rules for inequality and disequality are as follows, where $\succcurlyeq$ denotes either $>$ or $\ge$.
\[
\begin{array}{c}
\prftree[r]{dI\labineq}{P\rightarrow \dot{f}\ge 0}{\inv{P}{\dot{\boldsymbol{x}}=\boldsymbol{e}}{f\succcurlyeq 0}} \quad
\prftree[r]{dI\labneq}{P\rightarrow \dot{f}=0}{\inv{P}{\dot{\boldsymbol{x}}=\boldsymbol{e}}{f\neq0}}
\end{array}
\]

\vspace{-5mm}
\subsubsection{Differential Cut}

The differential cut rule (dC) inserts an intermediate invariant to be proved, and afterwards permits the use of this invariant to show further invariants. In contrast to \dL, it is not possible to record previously proved invariants as conjuncts in the domain of ODE commands. Instead we place them in the premise of the invariant triple. Indeed this is the primary motivation for introducing the concept of invariant triples.
\[
\prftree[r]{dC}{\inv{P}{\dot{\boldsymbol{x}}=\boldsymbol{e}}{Q_1}}{\inv{P\wedge Q_1}{\dot{\boldsymbol{x}}=\boldsymbol{e}}{Q_2}}{\inv{P}{\dot{\boldsymbol{x}}=\boldsymbol{e}}{Q_1\wedge Q_2}}
\]

The (dC) rule can be used multiple times to show conjunction of more than two invariants. For example, if we wish to show three invariants $Q_1,Q_2,Q_3$ in that order, first apply the (dC) rule with $Q_1$ and $Q_2$ to obtain $\inv{P}{\dot{\boldsymbol{x}}=\boldsymbol{e}}{Q_1\wedge Q_2}$, then apply the (dC) rule again to obtain the conclusion.

\vspace{-5mm}
\subsubsection{Differential Ghost}

The differential ghost rule (dG) adds new variables satisfying some differential equations to help prove the Hoare triple of the original differential equations.
\[
\prftree[r]{dG}{\inv{\overline{D}}{\dot{\boldsymbol{x}}=\boldsymbol{e}, \dot{\boldsymbol{y}}=\boldsymbol{f}(\boldsymbol{x},\boldsymbol{y})}{I}}{\partial D \wedge I \rightarrow Q}{\hoare{(D\rightarrow \exists \boldsymbol{y}.\,I)\wedge (\neg D\rightarrow Q)}{\ode{\dot{\boldsymbol{x}}=\boldsymbol{e}}{D}}{Q}}
\]
Here, $\boldsymbol{y}$ are fresh variables that do not occur in $\ode{\dot{\boldsymbol{x}}=\boldsymbol{e}}{D}$ or $Q$, and $\boldsymbol{f}(\boldsymbol{x},\boldsymbol{y})$ satisfies the Lipschitz condition.

\vspace{-5mm}
\subsubsection{Barrier Certificate}

The barrier certificate rule (bc) concludes invariants from the definition of barrier certificate.
\[
\prftree[r]{bc}{P\wedge f=0 \rightarrow \dot{f}>0}
{\inv{P}{\dot{\boldsymbol{x}}=\boldsymbol{e}}{f\succcurlyeq 0}}
\]

\vspace{-5mm}
\subsubsection{Darboux}
The Darboux rule (dbx) exploits properties of Darboux invariants, which are inspired by Darboux polynomials. Darboux equality and inequality rules are as follows.
\[
\begin{array}{c}
\prftree[r]{dbx\labeq}{P \rightarrow \dot{f}=gf}{\inv{P}{\dot{\boldsymbol{x}}=\boldsymbol{e}}{f=0}} \quad
\prftree[r]{dbx\labineq}{P \rightarrow \dot{f} \ge gf}{\inv{P}{\dot{\boldsymbol{x}}=\boldsymbol{e}}{f \succcurlyeq 0}}
\end{array}
\]

\vspace{-5mm}
\subsection{Solution Rule}

The solution rule offers another way to conclude Hoare triples directly, independent of using the (dW) or (dG) rule followed by proving invariants. In the rule below, $\boldsymbol{e}$ is linear in $\boldsymbol{x}$, and $\boldsymbol{u}(t,\boldsymbol{x})$ is the unique solution to the differential equation $\dot{\boldsymbol{x}}=\boldsymbol{e}$ with symbolic initial value $\boldsymbol{x}$ (that is, $\frac{d\boldsymbol{u}(t,\boldsymbol{x})}{dt}=\boldsymbol{e}(\boldsymbol{u}(t,\boldsymbol{x}))$ and $\boldsymbol{u}(0,\boldsymbol{x})=\boldsymbol{x}$). Let $P'(\boldsymbol{x})$ denote the following predicate on the starting state $\boldsymbol{x}$:
\[ \forall t>0.\, (\forall 0\le\tau<t.\, D(\boldsymbol{u}(\tau,\boldsymbol{x})))\wedge \neg D(\boldsymbol{u}(t,\boldsymbol{x})) \rightarrow Q(\boldsymbol{u}(t,\boldsymbol{x})).\]
The solution rule for Hoare triples (sln) is:
\[
\prftree[r]{sln}{}{
\hoare{(D\rightarrow P')\wedge (\neg D\rightarrow Q)}{\ode{\dot{\boldsymbol{x}}=\boldsymbol{e}}{D}}{Q}
}
\]

\begin{rep}
The version for invariant judgements (slnI) is:
\[
\prftree[r]{slnI}{}{\inv{\overline{D}, \forall t > 0.\, (\forall 0 \leq \tau \le t.\, \overline{D}(u(\tau,x))) \rightarrow Q(u(t,x))}
{\dot{x}=e}
{Q}}
\]
\end{rep}




\section{Verification Condition Generation} \label{sec:verif}

The VC generation procedure operates on annotated sequential HCSP programs. For ODEs, there are two kinds of annotations: ghost variable (\textsf{gvar}) and invariant annotations (\textsf{ode\_inv}):
\[
\begin{array}{rll}
\mathsf{gvar} &::=&  \mbox{ghost } z\ (\dot{z}=f(\boldsymbol{x},z)) \\
\mathsf{ode\_inv} &::=& [I]~|~[I]\ \{\mbox{dbx } g\}~|~[I]\ \{\mbox{bc}\}
\end{array}
\]
Here, `$\mbox{ghost } z\ (\dot{z}=f(\boldsymbol{x},z))$' denotes a ghost variable $z$ following the ODE $\dot{z}=f(\boldsymbol{x},z)$, where $f$ must be linear in $z$ to ensure global Lipschitz condition. The annotation $[I]$ denotes showing invariant $I$ using the (dI) rule. The annotation $[I]\ \{\mbox{dbx } g\}$ denotes showing an invariant using the (dbx) rule, with $g$ being the optional cofactor. The annotation $[I]\ \{\mbox{bc}\}$ denotes using the (bc) rule.

The syntax for annotated sequential HCSP programs is:
\[
\begin{array}{lll}
\mathcal{S},\mathcal{T} &::=& \skipcmd~|~
x := e~|~
x := * (B) ~|~
\mathcal{S}; \mathcal{T} ~|~
\itecmd{B}{\mathcal{S}}{\mathcal{T}} ~|~ \\
&& \mathcal{S} \ichoice \mathcal{T} ~|~
\loopcmd{\mathcal{S}} \mbox{ invariant } [I_1] \dots [I_n] ~|~ \\
&& \ode{\dot{\boldsymbol{x}}=\boldsymbol{e}}{D} \mbox{ invariant } \mathsf{gvar}_1\dots\mathsf{gvar}_k,\mathsf{ode\_inv}_1 \dots \mathsf{ode\_inv}_n ~|~ \\
&& \ode{\dot{\boldsymbol{x}}=\boldsymbol{e}}{D} \mbox { solution }
\end{array}
\]
The only addition to the syntax of HCSP is that each loop is followed by a list of invariants $I_1,\dots,I_n$, and each ODE is either followed by a list of ghost variable declarations and a list of invariant annotations, each of which specify an invariant to be proved using one of (dI), (dbx), or (bc) rules, or followed by the annotation ``solution'' to indicate that the (sln) rule is to be used.

To generate the necessary VCs for a given Hoare triple, we devised a procedure using weakest preconditions \cite{dijkstra-1975-guarded, dijkstra-1976-discipline}.
To be able to refer to preconditions and VCs individually,
we consider sets of conditions instead of composing predicates by $\land$.

Given a Hoare triple $\hoare{P_1\land \dots \land P_m}{\mathcal{S}}{Q_1\land\dots\land Q_n}$ to verify,
we define the set of all VCs to be
\[
\begin{array}{llr}
&\mathrm{VC}(\hoare{P_1\land \dots\land P_m}{\mathcal{S}}{Q_1\land\dots\land Q_n})=\\
&\quad\{P_1 \land \cdots \land P_m \rightarrow R \mid R \in \pre(\mathcal{S},\{Q_1,\dots,Q_n\})\} \cup{} &\qquad\qquad\qquad\qquad\ \ \vctag{pre} \\
&\quad\{\tilde{P}_1 \land \cdots \land \tilde{P}_{\tilde{m}} \rightarrow R
\mid R \in \vc(\mathcal{S}, \{Q_1,\dots,Q_n\}) \} 
&\vctag{vc}
\end{array}
\]
where $\tilde{P}_1, \ldots, \tilde{P}_{\tilde{m}}$ is the subset of 
the preconditions $P_1, \ldots, P_m$ whose variables are
never reassigned in $\mathcal{S}$,
and the functions $\pre$ and $\vc$ are defined below.

Given an annotated program $\mathcal{S}$ and
a set $\{Q_1, \dots, Q_n\}$ of postconditions,
we denote the set of derived
preconditions as $\pre(\mathcal{S}, \{Q_1, \dots, Q_n\})$,
defined as follows.
%
%
\[
\begin{array}{lr}
    \pre(\mathcal{S}, \{Q_1,\dots,Q_n\}) =
    \pre(\mathcal{S},Q_1) \cup \dots \cup \pre(\mathcal{S},Q_n) 
    &\vctag{pre-multi}
    \\
    \pre(\skipcmd, Q) = 
    Q
    &\vctag{pre-skip}
    \\
    \pre(x := e, Q) = Q[e/x]
    &\vctag{pre-assn}
    \\
    \pre(\mathcal{S}; \mathcal{T}, Q) = \pre(\mathcal{S}, \pre(\mathcal{T}, Q)) 
    &\vctag{pre-seq}
    \\
    \pre(\mathtt{if}\ B_1\ \mathtt{then}\ \mathcal{S}_1\ \mathtt{else}\ \cdots\ \mathtt{if}\ B_{n-1}\ \mathtt{then}\ \mathcal{S}_{n-1}\ \mathtt{else}\ \mathcal{S}_n, Q) = 
    \\
    \quad\{\neg (B_1 \lor \cdots \lor B_{i-1})\land B_i \to P  \mid P\in \pre(\mathcal{S}_i, Q), 1 \leq i \leq n-1\} ~ \cup 
    &\vctag{pre-if}
    \\
    \quad\{\neg (B_1 \lor \cdots \lor B_{n-1})\to P
    \mid P\in \pre(\mathcal{S}_n, Q)\}
    &\vctag{pre-else}
    \\
    \pre(\mathcal{S}_1 \ichoice \cdots \ichoice \mathcal{S}_n, Q) = 
    \pre(\mathcal{S}_1, Q) \cup \cdots \cup \pre(\mathcal{S}_n, Q)
    &\vctag{pre-choice}
    \\
    \pre(x := *\>(B), Q) = 
     B[y/x] \rightarrow Q[y/x]
      \text{ for a fresh variable $y$}
    &\vctag{pre-nassn}
    \\
    \pre(\loopcmd{\mathcal{S}} \mbox{ invariant } [I_1]\dots [I_n], Q) = 
    \{I_j~|~ 1\le j\le n\}
    &\vctag{pre-loop}
    \\
    \pre(\ode{\dot{\boldsymbol{x}}=\boldsymbol{e}}{D} \mbox{ invariant } \mathsf{gvar}_1\dots\mathsf{gvar}_k,\mathsf{ode\_inv}_1\dots\mathsf{ode\_inv}_n, Q) = \\
    \quad P_{\mathrm{skip}} \cup P_{\mathrm{init}} \\
    \pre(\ode{\dot{\boldsymbol{x}}=\boldsymbol{e}}{D} \mbox{ solution}) = P_{\mathrm{skip}} \cup P_{\mathrm{sln}}
\end{array}
\]
where
\[
\begin{array}{llr}
    & P_{\mathrm{skip}} = \{\neg D \rightarrow Q\}
    & \qquad\qquad\qquad\vctag{pre-dWG-skip} \\
    & P_{\mathrm{init}} = \{D \rightarrow \exists z_1, \ldots, z_k.\> I_1 \land \cdots \land I_n \} 
    \quad\text{if $k>0$}
    & \vctag{pre-dG-init} \\
    & P_{\mathrm{init}} = \{D \rightarrow I_j, \mid 1 \le j \le n \} \text{ otherwise}
    & \vctag{pre-dW-init} \\
    & P_{\mathrm{sln}} = \{D \rightarrow (\forall t>0.\, (\forall 0\le\tau<t.\, D(\boldsymbol{u}(\tau,\boldsymbol{x})))\wedge {} \\
    & \phantom{P_{\mathrm{sln}} = \{} \neg D(\boldsymbol{u}(t,\boldsymbol{x}))\rightarrow Q(\boldsymbol{u}(t,\boldsymbol{x}))) \}
    & \vctag{pre-sln}
\end{array}
\]
where $z_1,\dots,z_k$ are the ghost variables provided in $\mathsf{gvar}_1\dots\mathsf{gvar}_k$, and $I_1, \ldots, I_n$ are the invariants provided in $\mathsf{ode\_inv}_1,\dots,\mathsf{ode\_inv}_n$.
If the user chooses the (sln) rule, we verify that $\boldsymbol{e}$ is linear in $\boldsymbol{x}$ and compute the unique solution $\boldsymbol{u}(\tau,\boldsymbol{x})$ to the ODE with symbolic initial value $\boldsymbol{x}$.




Given an annotated program $\mathcal{S}$ and
a set $\{Q_1 \dots, Q_n\}$ of postconditions,
we denote the set of internal VCs as $\vc(\mathcal{S}, \{Q_1, \dots, Q_n\})$, defined as follows.
\[
\begin{array}{lr}
\vc(\mathcal{S}, \{Q_1,\dots,Q_n\}) =
\vc(\mathcal{S},Q_1) \cup \dots \cup \vc(\mathcal{S},Q_n)
& \vctag{vc-multi}
\\
\vc(\skipcmd, Q) = \emptyset
& \vctag{vc-skip}
\\
\vc(x := e, Q) = \emptyset
& \vctag{vc-assn}
\\
\vc(\mathcal{S};\mathcal{T}, Q) = \vc(\mathcal{S}, \pre(\mathcal{T}, Q)) \cup \vc(\mathcal{T}, Q)
& \vctag{vc-seq}
\\
\multicolumn{2}{l}{\vc(\mathtt{if}\ B_1\ \mathtt{then}\ \mathcal{S}_1\ \mathtt{else\ if}\ \cdots\ \mathtt{else\ if}\ B_{n-1}\ \mathtt{then}\ \mathcal{S}_{n-1}\ \mathtt{else}\ \mathcal{S}_n, Q) =}\\
\quad\vc(\mathcal{S}_1,Q) \cup \dots \cup \vc(\mathcal{S}_n,Q)
& \vctag{vc-ite}
\\
\vc(\mathcal{S}_1 \ichoice \cdots \ichoice \mathcal{S}_n, Q) = \vc(\mathcal{S}_1,Q) \cup \dots \cup \vc(\mathcal{S}_n,Q)
& \vctag{vc-choice}
\\
\vc(x := *\>(B), Q) = \emptyset
& \vctag{vc-nassn}\\
\vc(\loopcmd{\mathcal{S}} \mbox{ invariant } [I_1]\dots[I_n], Q) = \\
\quad\vc(\mathcal{S}, \{I_1,\dots,I_n\}) \cup {}
& \vctag{vc-loop-body}
\\
\quad \{(I_1 \land \cdots \land I_n) \rightarrow Q\} \cup {}
& \vctag{vc-loop-exit}
\\
\quad \{(I_1 \land \cdots \land I_n) \rightarrow P
\mid P \in \pre(\mathcal{S}, \{I_1,\dots,I_n\}) \}
& \vctag{vc-loop-maintain}\\
\multicolumn{2}{c}{\vc(\ode{\dot{\boldsymbol{x}}=\boldsymbol{e}}{D} \mbox{ invariant } \mathsf{gvar}_1\dots \mathsf{gvar}_m, \mathsf{ode\_inv}_1\dots\mathsf{ode\_inv}_n, Q) = C_\mathrm{exec} \cup C_\mathrm{dC}} \\
\vc(\ode{\dot{\boldsymbol{x}}=\boldsymbol{e}}{D} \mbox{ solution}, Q) = \emptyset
\end{array}
\]
where we set $C_\mathrm{exec} = \emptyset$ if the only invariant is $\mathsf{false}$,
or else
\[
\begin{array}{llr}
& C_\mathrm{exec} = \{I_1 \land \dots \land I_n \land \partial D \rightarrow Q\}
&\vctag{vc-dWG-exec} \\
& C_\mathrm{dC} = \{ I_1 \land \cdots \land I_{j-1} \to R \mid R \in \vc(\ode{\dot{\boldsymbol{x}}=\boldsymbol{e}}{D}, \mathsf{ode\_inv}_j, Q), \\
& \qquad\qquad\qquad\qquad\qquad\qquad\quad 1 \le j \le n\}
& \vctag{vc-dC}
\end{array}
\]
Here, $I_1, \dots, I_n$ are the invariants
provided in $\mathsf{ode\_inv}_1,\dots,\mathsf{ode\_inv}_n$. If no invariants are specified, we set 
a single invariant $I_1 = \mathsf{true}$ by default.
We write
$\vc(\ode{\dot{\boldsymbol{x}}=\boldsymbol{e}}{D}, \mathsf{ode\_inv}_j, Q)$ for the VC generated from annotation $\mathsf{ode\_inv}_j$, defined as follows.
\[
\begin{array}{llr}
&\vc(\ode{\dot{\boldsymbol{x}}=\boldsymbol{e}}{D}, [\mathsf{true}], Q) = \emptyset
 & \vctag{vc-true}
\\
&\vc(\ode{\dot{\boldsymbol{x}}=\boldsymbol{e}}{D}, [\mathsf{false}], Q) = \emptyset
& \vctag{vc-false}
\\
&\vc(\ode{\dot{\boldsymbol{x}}=\boldsymbol{e}}{D}, [f = 0], Q) = \{\overline{D}\rightarrow \dot{f} = 0\}
& \vctag{vc-dI1}
\\
&\vc(\ode{\dot{\boldsymbol{x}}=\boldsymbol{e}}{D}, [f \succcurlyeq 0], Q) = \{\overline{D}\rightarrow \dot{f} \geq 0\}
& \vctag{vc-dI2}
\\
&\vc(\ode{\dot{\boldsymbol{x}}=\boldsymbol{e}}{D}, [f \ne 0], Q) = \{\overline{D}\rightarrow \dot{f} = 0\}
& \vctag{vc-dI3}
\\
&\vc(\ode{\dot{\boldsymbol{x}}=\boldsymbol{e}}{D}, [f=0]\ \{\mbox{dbx } g\}, Q) = \{\overline{D}\rightarrow \dot{f}=gf\}
& \vctag{vc-dbx1}
\\
&\vc(\ode{\dot{\boldsymbol{x}}=\boldsymbol{e}}{D}, [f\succcurlyeq 0]\ \{\mbox{dbx } g\}, Q) =
\{\overline{D}\rightarrow \dot{f}\ge gf\}
& \vctag{vc-dbx2}
\\
&\vc(\ode{\dot{\boldsymbol{x}}=\boldsymbol{e}}{D}, [f\succcurlyeq 0]\ \{\mbox{bc}\}, Q) = \{\overline{D}\wedge f=0\rightarrow \dot{f}>0 \}
 & \qquad\qquad\qquad \vctag{vc-bc}
\end{array}
\]
All Lie derivatives are computed with respect to
$\dot{\boldsymbol{x}}=\boldsymbol{e}$ and the equations given in
$\mathsf{gvar}_1\dots \mathsf{gvar}_m$.
For the (dbx) rule, if no cofactor $g$ is provided, we attempt to compute the cofactor automatically. Specifically, in the case of an equality invariant, this reduces to simplifying $\dot{f} / f$ into polynomial form. In the case of an inequality invariant, we attempt to find a polynomial quotient of $\dot{f}$ and $f$ with a non-negative remainder. 

\newcommand\thmvcgsound{A Hoare triple $\hoare{P_1\wedge \dots\wedge P_m}{\mathcal{T}}{Q_1\wedge\dots\wedge Q_n}$
holds if all conditions in
$\mathrm{VC}(\hoare{P_1\wedge \dots\wedge P_m}{\mathcal{T}}{Q_1\wedge \dots\wedge Q_n})$
hold.}

\begin{theorem} \label{thm:vcg-sound}
\thmvcgsound
\end{theorem}
\begin{proof}
We give the full proof in Appendix~\ref{sec:proof-vcg}. 
In short,
we proceed by structural induction on $\mathcal{T}$.
The difficult case is when $\mathcal{T}$ is an ODE. If the ODE is annotated to use the solution rule, we use the VC stemming from the precondition \vctag{pre-sln}. Otherwise, we employ the (dG) rule or the (dW) rule depending on if ghost variables are specified. The VCs stemming from \vctag{pre-dWG-skip} and \vctag{pre-dW-init} or \vctag{pre-dG-init} show that the rule (dW) or (dG) is applicable. The condition
\vctag{vc-dWG-exec} 
discharges the right premise of the (dW) or (dG) rule.

For the left premise
$\inv{\overline{D}}{\dot{\boldsymbol{x}}=\boldsymbol{e},
\dot{\boldsymbol{y}}=\boldsymbol{f}(\boldsymbol{x},\boldsymbol{y})}{I_1\land\cdots \land I_k}$,
(without $\boldsymbol{y}$ if using the (dW) rule), we repeatedly apply the (dC) rule to isolate each invariant $I_i$.
For each step
$\inv{\overline{D} \land I_1 \land \dots \land I_{i-1}}{\dot{\boldsymbol{x}}=\boldsymbol{e},
\dot{\boldsymbol{y}}=\boldsymbol{f}(\boldsymbol{x},\boldsymbol{y})}{I_i}$, depending on the rule specified in the annotation, we apply the (dI) rule, (dbx) rule or (bc) rule, using the corresponding VCs as premises.
\qed
\end{proof}

\section{Labels}
\label{sec:labels}

VCs generated by the procedure in Section~\ref{sec:verif} will be proved using Z3 or Wolfram Engine. In this section, we introduce a labeling mechanism to store which solver is used for each VC, in a way that is robust to minor modifications of the program or its annotations.

As indicated in Section~\ref{sec:verif}, the generation of a VC starts from a postcondition or invariant and proceeds bottom up through the program. We call the postcondition or invariant at the beginning of this process the \emph{conclusion assertion} of the VC. We associate each VC to its conclusion assertion. Labels are used to distinguish between multiple VCs from the same conclusion assertion. They can arise for the following reasons:
\begin{itemize}
    \item Loop and ODE invariants produce VCs for showing that they initially hold and for showing that they are maintained by the loop or ODE.
    \item If-then-else and nondeterministic choice produce multiple preconditions, at least one for each branch.\footnote{A large program consisting of multiple if-then-else commands can lead to an undesired blow up of the number of VCs. For example, a program constructed by ten if-then-else commands as sequential components results in more than $2^{10}$ VCs.}
    \item Each ODE produces preconditions for both the case when the domain $D$ holds initially, and for when $D$ does not hold.
\end{itemize}

A label consists of two parts: a category label and a branch label. The category label is either empty or one of ``init'', ``maintain'', ``init\_all''. The branch label is a list, separated by ``.'', of either ``skip'', ``exec'', or $n(b)$ where $n$ is a positive integer and $b$ is a branch label itself. We write $n$ instead of $n()$ when the inner branch label is empty.






\subsubsection{Category Labels}

Category labels use ``init" (``init\_all") and ``maintain" to distinguish between VCs with loop or ODE invariants as conclusion assertions. For loops, the VCs for showing the invariant holds initially are labeled ``init", and the VCs that result from showing the invariant is maintained by the loop are labeled ``maintain'' (when there are nested loops or ODEs in the loop body, multiple VCs are computed in the loop body, this applies only to those with the invariant as conclusion assertion).

For an ODE, the VC coming from \vctag{pre-dW-init} (resp. \vctag{pre-dG-init}) are labeled ``init'' (resp. ``init\_all''). The VCs coming from \vctag{vc-dC}, for showing each invariant is maintained during evolution, are labeled ``maintain''.

The category label is empty in all other cases.

\begin{rep}
two sorts of VCs with invariants as their conclusion assertions, which are used in loops and ODEs. One kind of VC is that the invariants are satisfied at the initial state, which is labeled ``init"(``init\_all"). The other is that the invariants are maintained after executing the loop body or during the evolution of ODE, which is labeled ``maintain".

Specifically, for a loop $S*$ with invariants $I_1, \dots, I_n$, the precondition $I_i$
\vctag{pre-loop} is labeled with ``init" for each invariant $I_i$. The VC $I_1 \land \cdots \land I_n \rightarrow P$ for each $P \in\pre(S, I_i)$ \vctag{vc-loop-maintain} is labeled ``maintain" if $\vc(S, I_i) = \emptyset$. Otherwise, the VC in $\vc(S, I_i)$\vctag{vc-loop-body} with $I_i$ as its conclusion assertion will be labeled ``maintain" ($\vc(S, I_i) \neq \emptyset$ when $S$ contains loop or ODE commands). For example, S is a sequence command in the form of `$x:=e; L*$', where $L*$ is a loop with invariant $F$. To prove that $I_i$ is maintained after executing $S$, we have to show that $I_1 \land \cdots \land I_n \rightarrow F$, $F$ is maintained after executing $L$ (in $\vc(S, I_i)$) and $F\rightarrow I_i$ (in $\vc(S, I_i)$), in which $F\rightarrow I_i$ should be labeled ``maintain" for the invariant $I_i$.

For an ODE $\ode{\dot{\boldsymbol{x}}=\boldsymbol{e}}{D}$ with invariants $I_1, \dots, I_n$, if no ghost variables are introduced, the precondition $D \rightarrow I_i$ \vctag{pre-dW-init} is labeled ``init" for each invariant $I_i$. If ghost variables $y$ are introduced, the precondition will be $D \rightarrow \exists y.\,I$ \vctag{pre-dG-init}, which is labeled ``init\_all". The VC \vctag{vc-dC} generated for the invariant triples are labeled ``maintain" for each invariant $I_i$, which indicates that $I_i$ is maintained during the evolution of $\ode{\dot{\boldsymbol{x}}=\boldsymbol{e}}{D}$.

Other VCs or preconditions, if they are generated from 
\vctag{pre-loop}, \vctag{pre-dW-init} or \vctag{pre-dG-init}, will inherit the category labels from the precondition
they originally stem from. Otherwise, they do not have category labels by default.
\end{rep}

\subsubsection{Branch Labels}

Branch labels help to distinguish VCs generated by executing different branches of programs. 

The positive integer $n$ handles branches created by `$\itecmd{B}{\mathcal{S}_1}{\mathcal{S}_2}$' or `$\mathcal{S}_1 \ichoice \mathcal{S}_2$'. Each value of $n$ (starting from 1) corresponds to one branch. Sequence labels $b.b$ are used for sequences of such commands. For example, the branches for `$\mathcal{S}_1 \ichoice \mathcal{S}_2; \mathcal{S}_3 \ichoice \mathcal{S}_4$' have labels $1.1$, $1.2$, $2.1$ and $2.2$. Nested labels $n(b)$ are used for nested commands. For example, the branches for $\itecmd{B}{\mathcal{S}_1 \ichoice \mathcal{S}_2}{\mathcal{S}_3}$ have labels $1(1)$, $1(2)$ and 2, corresponding to $\mathcal{S}_1$, $\mathcal{S}_2$ and $\mathcal{S}_3$, respectively.

The labels ``skip" and ``exec" are used for branches of the ODE. The branch where the initial state does not satisfy domain $D$ is labeled ``skip''. The other branch, where the ODE is executed, is labeled ``exec''. They come from applying the rules \vctag{vc-dWG-skip} and \vctag{vc-dWG-exec}, respectively.

\begin{rep}
Sequence labels ($b.b$) concatenate branch labels of subcommands of sequence command with ``.". For example, for `$S_1 \ichoice S_2; S_3 \ichoice S_4$', the precondition or VC generated by executing $S_1$ and $S_3$ is labeled 1.1, while the precondition or VC generated by executing $S_1$ and $S_4$ is labeled 1.2. When the branch label of subcommand has no value, it will not be concatenated.

Nested labels ($z(b)$) complement $z$ when $S_1$ or $S_2$ in `$\itecmd{B}{S_1}{S_2}$' or `$S_1 \ichoice S_2$' also create branches. The branch label created by the subcommand is nested with $z$. For example, for the command `$\itecmd{B}{S_1 \ichoice S_2}{S_3}$', the precondition or VC generated by executing $S_1$ is labeled 1(1), which means executing the first branch of $\ichoice$ in the if-branch. The precondition or VC generated by executing $S_2$ is labeled 1(2), and the precondition or VC generated by executing $S_3$ is labeled 2.

Assignments and skip commands do not have branch labels to distinguish different branches because they only have one branch to execute.

Finally, category labels and branch labels together form the final labels for VCs. Some of the VCs are already labeled above, and others formed in $\mathit{original\ precondition} \rightarrow \mathit{derived\ precondition}$ will inherit the labels of the derived preconditions. Labels can be only category labels (branch labels) if there are no branch labels(category labels). And some VCs have no labels if they are the only VC generated from its conclusion assertion. 
\end{rep}

\begin{rep}
\begin{example} Consider the following Hoare triple in which the program contains sequence, if-then-else, ODE, assignment and skip commands. The conclusion assertions and labels of each VC is shown in Table~\ref{tab:label}, where $x\geq 0$~(post) represents the postcondition, and $x\geq 0$~(inv) represents the invariant.
\[
\begin{split}
 \{x \leq 1\}  
 &\itecmd{x\leq 0}{x:=-x}{skip};
t:=0; \\
 &\ode{\dot{x}=1, \dot{t}=1}{t<1}\ 
\mathtt{invariant}\ [x\geq 0] 
\{x \geq 0\}   
\end{split}
\]
\end{example}

\begin{table}[]
    \caption{VCs and their Corresponding Conclusion Assertions and Labels}
    \label{tab:label}
    \centering
    \begin{tabular}{c | c | c}
        \hline
        VC & Conclusion Assertion & Label \\ \hline
        
        $x\geq 0 \land t=1 \rightarrow x\geq 0$ & 
        $x \geq 0$~(post) &
        exec\\ \hline
        
        $x \leq 1 \rightarrow x \leq 0 \rightarrow 0 \geq 1 \rightarrow -x \geq 0$ &
        $x \geq 0$~(post) &
        1.skip \\ \hline
        
        $x \leq 1 \rightarrow x > 0 \rightarrow 0 \geq 1 \rightarrow x \geq 0$ &
        $x \geq 0$~(post) &
        2.skip \\ \hline
        
        $t \leq 1 \rightarrow 1 \geq 0$ &
        $x \geq 0$~(inv) &
        maintain \\ \hline
        
        $x \leq 1 \rightarrow x \leq 0 \rightarrow 0 < 1 \rightarrow -x \geq 0$ &
        $x \geq 0$~(inv) &
        init 1 \\ \hline
        
        $x \leq 1 \rightarrow x > 0 \rightarrow 0 < 1 \rightarrow x \geq 0$ &
        $x \geq 0$~(inv) &
        init 2\\ \hline
    \end{tabular}
    
\end{table}
\end{rep}


\begin{example}
\label{ex:1}
This example illustrates assignments, nondeterministic choice, and loops.
\[
\begin{array}{ll}
& \{x \le 0\} \\
& \quad x := -x; \\
& \quad (x := x + 1 \ichoice x := x + 2)* \\
& \quad \quad \mbox {invariant } [x \ge 0]; \\
& \quad x := x + 1 \\
& \{x \ge 1\}
\end{array}
\]
The computation starts at postcondition $x\geq 1$. Applying \vctag{pre-assn} and \vctag{vc-loop-exit}, we get the VC $x\geq 0 \rightarrow x+1\geq 1$. Applying \vctag{pre-loop}, the loop's precondition is $x \geq 0$. The whole program's precondition is $-x \geq 0$ by applying \vctag{pre-assn} again. The loop body yields the preconditions $x+1\geq 0$ and $x+2\geq 0$ by \vctag{pre-choice} and \vctag{pre-assn}. Then we get $x\geq 0 \rightarrow x+1\geq 0$ and $x\geq 0 \rightarrow x+2\geq 0$ by applying \vctag{vc-loop-maintain}.
The overall VCs and their labels are:
\end{example}

\begin{center}
\small
\begin{tabular}{c | c | c}
    \hline
    VC & Conclusion Assertion & Label \\ \hline
    
    $x \le 0 \rightarrow -x\ge 0$ & 
    $x \geq 0$~(inv) &
    init\\ \hline
    
    $x\ge 0 \rightarrow x + 1\ge 1$ &
    $x \geq 1$~(post) &
    $\epsilon$ \\ \hline
    
    $x\ge 0 \rightarrow x + 1\ge 0$ &
    $x \geq 0$~(inv) &
    maintain 1 \\ \hline
    
    $x\ge 0 \rightarrow x + 2\ge 0$ &
    $x \geq 0$~(inv) &
    maintain 2 \\ \hline
\end{tabular}
\end{center}
With conclusion assertions and labels, we can store the solver (default Z3) for each VC and reuse the solver despite of minor modifications of code. For example, if we choose Wolfram Engine to prove $x \le 0 \rightarrow -x\ge 0$, ``init: wolfram" will be annotated after the invariant $x \geq 0$. If we then change the second line from $x:=-x$ into $x:=-2*x$, resulting in a different VC $x \le 0 \rightarrow -2 * x \ge 0 $, the solver of the VC is still Wolfram Engine.
\begin{example}
\label{ex:2}
This example illustrates non-deterministic assignments and ODEs (\#4 of KeYmaera X's basic benchmarks):
\[
\begin{array}{ll}
& \{x\ge 0\} \\
& \quad x := x+1; t := *\>(t\ge 0); \\
& \quad \ode{\dot{t}=-1,\dot{x}=2}{t>0}\mbox{ invariant } [x\ge 1] \\
& \{x\ge 1\}
\end{array}
\]
The computation of $\pre$ starts at postcondition $x\ge 1$. By $\vctag{pre-dWG-skip}$ and $\vctag{pre-dW-init}$, the ODE's preconditions are $\neg t>0\rightarrow x\ge 1$ and $t>0\rightarrow x\ge 1$. By $\vctag{pre-nassn}$ and $\vctag{pre-assn}$, the whole program's preconditions are $t_1\ge 0\rightarrow t_1>0\rightarrow x\ge 1$ and $t_1\ge 0\rightarrow \neg t_1>0\rightarrow x+1\ge 1$. The VCs $x\ge 1\wedge t=0\rightarrow x\ge 1$ and $t\ge 0\rightarrow 2\ge 0$ come from $\vctag{vc-dWG-exec}$ and $\vctag{vc-dI2}$, respectively. The overall list of VCs is:
\begin{center}
\small
\begin{tabular}{c | c | c}
    \hline
    VC & Conclusion Assertion & Label \\ \hline
    
    $x\ge 0\rightarrow t_1\ge 0\rightarrow t_1>0 \rightarrow x+1\ge 1$ & 
    $x \geq 1$~(inv) &
    init\\ \hline
    
    $x\ge 0\rightarrow t_1\ge 0\rightarrow \neg t_1 > 0 \rightarrow x+1\ge 1$ &
    $x \geq 1$~(post) &
    skip \\ \hline
    
    $x\ge 1\wedge t=0\rightarrow x\ge 1$ &
    $x \geq 1$~(post) &
    exec \\ \hline
    
    $t\ge 0\rightarrow 2\ge 0$ &
    $x \geq 1$~(inv) &
    maintain  \\ \hline
\end{tabular}
\end{center}
\end{example}

\begin{example}
\label{ex:3}
Finally, we consider an example with multiple ghost variables (\#18 of KeYmaera X's basic benchmarks):
\[
\begin{array}{ll}
& \{x \ge 0\} \\
& \quad t := *\>(t \ge 0); \ode{\dot{x}=x, \dot{t}=-1}{t > 0} \\
& \quad \mbox {invariant ghost } y\ (\dot{y}=-y) 
  \ \mbox {ghost } z\ (\dot{z} = z/2) \\
& \quad\quad [xy\ge 0] \ [yz^2 = 1] \\
& \{x \ge 0\}
\end{array}
\]
\begin{center}
\small
    \begin{tabular}{c | c | c}
        \hline
        VC & Conclusion Assertion & Label \\ \hline
        
        $x\ge 0\rightarrow t_1\ge 0\rightarrow t_1>0 \rightarrow \exists y\,z.\, xy\ge 0\wedge yz^2=1$ & 
        invariants &
        init\_all \\ \hline
        
        $x\ge 0\rightarrow t_1\ge 0\rightarrow \neg t_1 > 0 \rightarrow x\ge 0$ &
        $x \geq 0$~(post) &
        skip \\ \hline
        
        $xy\ge 0\wedge yz^2=1\wedge t=0 \rightarrow x\ge 0$ &
        $x \geq 0$~(post) &
        exec \\ \hline
        
        $t\ge 0\rightarrow x\cdot(-y)+xy\ge 0$ &
        $xy\ge 0$~(inv) &
        maintain \\ \hline
        
        $xy\ge 0\rightarrow t\ge 0\rightarrow yz(z/2) + (y(z/2) + (-y)z)z = 0$ &
        $yz^2=1$~(inv) &
        maintain \\ \hline
    \end{tabular}
\end{center}
The first VC comes from $\vctag{pre-dG-init}$. The remaining VCs are similar to Example~\ref{ex:2}, except that there is one VC for maintaining each invariant. When verifying the second invariant, the (dC) rule allows us to assume the first invariant.
\end{example}

\section{Highlighting}
\label{sec:highlighting}

In this section, we explain the highlighting mechanism we devised to help the user understand how each VC is derived from the program.
Essentially, when the user hovers over a VC, we highlight all parts
of the program that contribute to the computation of the VC,
including commands, assertions and domain constraints.

We highlight any assertion that contributes to the VC. In particular, invariants of an ODE that are already proved will be highlighted when proving the next invariant because they are added as assumptions in \vctag{vc-dC}. Preconditions whose variables are never reassigned will be highlighted because they are added as assumptions in \vctag{vc}.

Domain constraints of ODEs will be highlighted if they are used in the VC (e.g. the domain constraint $D$ in the VC generated by \vctag{vc-dWG-exec}).

Atomic commands are highlighted if they are traversed during VC generation. ODE commands are highlighted for VCs computed by \vctag{vc-dC} or \vctag{pre-sln}. For if-then-else and nondeterministic choice, only the branch that is actually traversed during VC generation will be considered for highlighting.

\begin{rep}
For example, in the following triple,
\[
\hoare{x \geq 0}{\itecmd{x < 1}{x:=x+1\ichoice x:=x+2}{\skipcmd}}{x\geq 1}
\]
the generation of VC `$x \geq 0 \rightarrow x < 1 \rightarrow x+1\geq 1$' traverses `$x:=x+1$' in the if-branch. Therefore, `$x:=x+1$' will be highlighted rather than the whole `$x:=x+1 \ichoice x:=x+2$' in the if-branch.
\end{rep}

Fig.\ \ref{fig:ex1} and \ref{fig:ex2} show the highlighting for some the VCs from Examples \ref{ex:1} and \ref{ex:2}.




\newcommand\highlight[1]{\adjustbox{bgcolor=lightgray}{\strut #1}}
\begin{figure}[tbh]
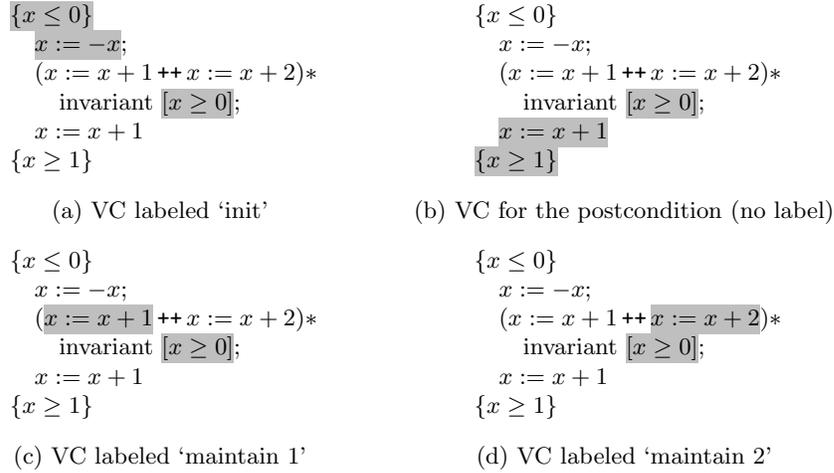

    \begin{subfigure}{0.5\textwidth}
    \[
    \begin{array}{ll}
    & \highlight{$\{x \le 0\}$} \\
    & \quad \highlight{$x := -x$}; \\
    & \quad (x := x + 1 \ichoice x := x + 2)* \\
    & \quad \quad \mbox {invariant } \highlight{$[x \ge 0]$}; \\
    & \quad x := x + 1 \\
    & \{x \ge 1\}
    \end{array}
    \]
    \vspace{-1mm}
    \caption{VC labeled `init'}
    \vspace{2mm}
    \end{subfigure}
    \begin{subfigure}{0.5\textwidth}
    \[
    \begin{array}{ll}
    & \{x \le 0\} \\
    & \quad x := -x; \\
    & \quad (x := x + 1 \ichoice x := x + 2)* \\
    & \quad \quad \mbox {invariant } \highlight{$[x \ge 0]$}; \\
    & \quad \highlight{$x := x + 1$} \\
    & \highlight{$\{x \ge 1\}$}
    \end{array}
    \]
    \vspace{-1mm}
    \caption{VC for the postcondition (no label)}
    \vspace{2mm}
    \end{subfigure}

    \begin{subfigure}{0.5\textwidth}
    \[
    \begin{array}{ll}
    & \{x \le 0\} \\
    & \quad x := -x; \\
    & \quad (\highlight{$x := x + 1$} \ichoice x := x + 2)* \\
    & \quad \quad \mbox {invariant } \highlight{$[x \ge 0]$}; \\
    & \quad x := x + 1 \\
    & \{x \ge 1\}
    \end{array}
    \]
    \vspace{-1mm}
    \caption{VC labeled `maintain 1'}
    \vspace{-1mm}
    \end{subfigure}
    \begin{subfigure}{0.5\textwidth}
    \[
    \begin{array}{ll}
    & \{x \le 0\} \\
    & \quad x := -x; \\
    & \quad (x := x + 1 \ichoice \highlight{$x := x + 2$})* \\
    & \quad \quad \mbox {invariant } \highlight{$[x \ge 0]$}; \\
    & \quad x := x + 1 \\
    &  \{x \ge 1\}
    \end{array}
    \]
    \vspace{-1mm}
    \caption{VC labeled `maintain 2'}
    \vspace{-1mm}
    \end{subfigure}
    \caption{Highlighting for the four VCs in Example~\ref{ex:1} 
    }
    \label{fig:ex1}
\end{figure}




\begin{figure}[tbh]
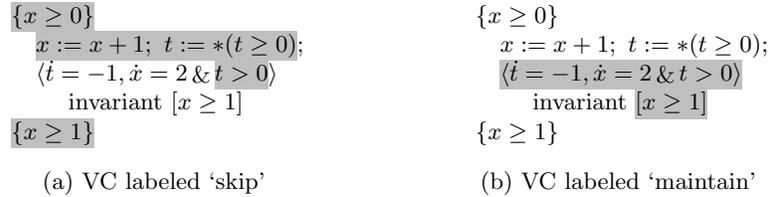

    \begin{subfigure}{0.5\textwidth}
    \[
    \begin{array}{ll}
    & \highlight{$\{x\ge 0\}$} \\
    & \quad \highlight{$x := x+1;\ t := * (t\ge 0)$}; \\
    & \quad \ode{\dot{t}=-1,\dot{x}=2}{\highlight{$t>0$}}\\
    &\quad\quad\mbox{ invariant } [x\ge 1] \\
    & \highlight{$\{x\ge 1\}$}
    \end{array}
    \]
    \vspace{-1mm}
    \caption{VC labeled `skip'}
    \vspace{-1mm}
    \end{subfigure}
    \begin{subfigure}{0.5\textwidth}
    \[
    \begin{array}{ll}
    & \{x\ge 0\} \\
    & \quad x := x+1;\ t := * (t\ge 0); \\
    & \quad \highlight{$\ode{\dot{t}=-1,\dot{x}=2}{t>0}$}\\
    &\quad\quad\mbox{ invariant } \highlight{$[x\ge 1]$} \\
    & \{x\ge 1\}
    \end{array}
    \]
    \vspace{-1mm}
    \caption{VC labeled `maintain'}
    \vspace{-1mm}
    \end{subfigure}
    \caption{Highlighting for the two of the VCs in Example~\ref{ex:2}}
    \label{fig:ex2}
\end{figure}

\begin{rep}
\begin{example}
\[
\hoare{x\geq 0 \wedge y \geq 0 }
{t:=*(t\geq 0); \ode{\dot{x}=y,\dot{y}=1, \dot{t}=-1}{t>0}\ \mathit{invariant}\ [y \geq 0][x\geq 0]}
{x\geq 0}
\]

\begin{table}[]
    \centering
    \caption{VCs and their Corresponding Highlighting Parts}
    \begin{tabular}{c | c}
        \hline
        VC & Highlighting \\ \hline
        $y \geq 0 \wedge x \geq 0 \wedge t =0 \rightarrow x \geq 0$ &
        $y \geq 0$, $x \geq 0$,$t>0$ $x \geq 0$ \\ \hline
        $t \geq 0 \rightarrow 1\geq 0$ &
        $\ode{\dot{x}=y,\dot{y}=1, \dot{t}=-1}{t>0}$, $y\geq0$ \\ \hline
        $t \geq 0 \wedge y\geq 0 \rightarrow y\geq 0$ &  $\ode{\dot{x}=y,\dot{y}=1, \dot{t}=-1}{t>0}$, $y\geq 0$, $x\geq 0$ \\ \hline
        $x\geq 0 \wedge y \geq 0  \wedge t_1\geq 0 \wedge t_1 > 0 \rightarrow x \geq 0$ &  $x\geq 0$, $y\geq 0$, $t:=*(t\geq 0)$, $x \geq 0$ \\ \hline
        $x\geq 0 \wedge y \geq 0  \wedge t_1\geq 0 \wedge t_1 > 0 \rightarrow y \geq 0$ &  $x\geq 0$, $y\geq 0$, $t:=*(t\geq 0)$, $y \geq 0$ \\ \hline
        $x\geq 0 \wedge y \geq 0  \wedge t_1\geq 0 \wedge t_1 \leq 0 \rightarrow x \geq 0$ &  $x\geq 0$, $y\geq 0$, $t:=*(t\geq 0)$, $x \geq 0$ \\ \hline
        
    \end{tabular}
    \label{tab:highlight}
\end{table}
\end{example}
\end{rep}

\section{Implementation and Evaluation}
\label{sec:evaluation}
In this section, we present the implementation of \textsf{HHLPy} and evaluate it on Simulink/Stateflow models and on KeYmaera X benchmarks. All verified examples are available online, coming with the tool.

\subsection{Implementation}

Fig.~\ref{fig:arch} shows the architecture of the tool. The user inputs HCSP programs and annotations in the editor (the HCSP programs can also come from translation of Simulink/Stateflow models). The core \textsf{HHLPy} engine then parses the input and generates VCs. The user interface displays the VCs and allows users to choose a solver for each VC. The solver will be invoked, with the results displayed to the user interface. The backend of \textsf{HHLPy} is implemented in Python, and the graphical user interface is implemented using JavaScript. A screenshot of the user interface is shown in Fig.~\ref{fig:screenshot}.
\begin{figure}[h!]
    \centering
    \includegraphics[width=\textwidth]{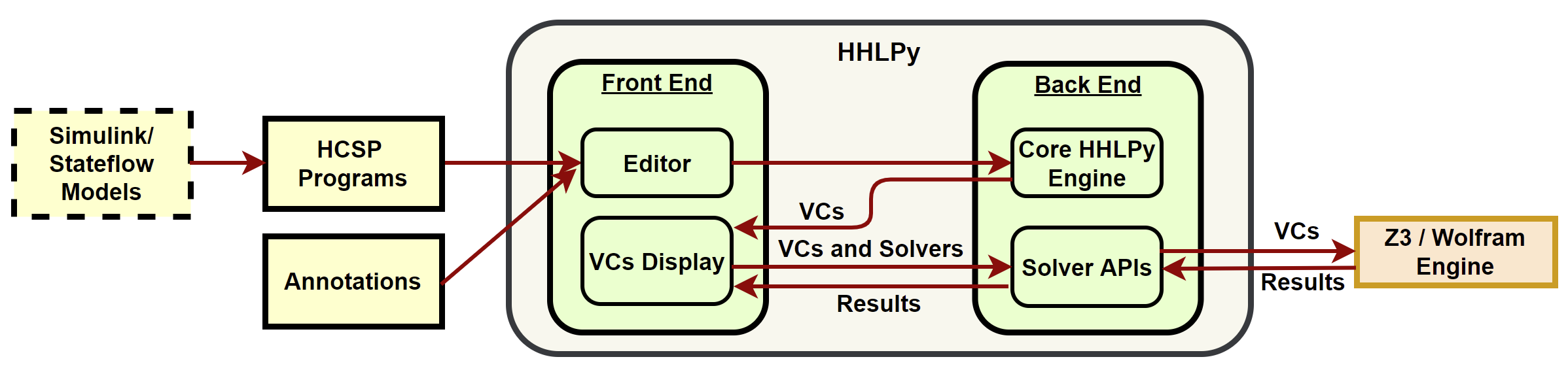}
    \caption{Architecture of HHLPy}
    \label{fig:arch}
\end{figure}

\begin{figure}[h!]
    \centering
    \includegraphics[scale=0.3]{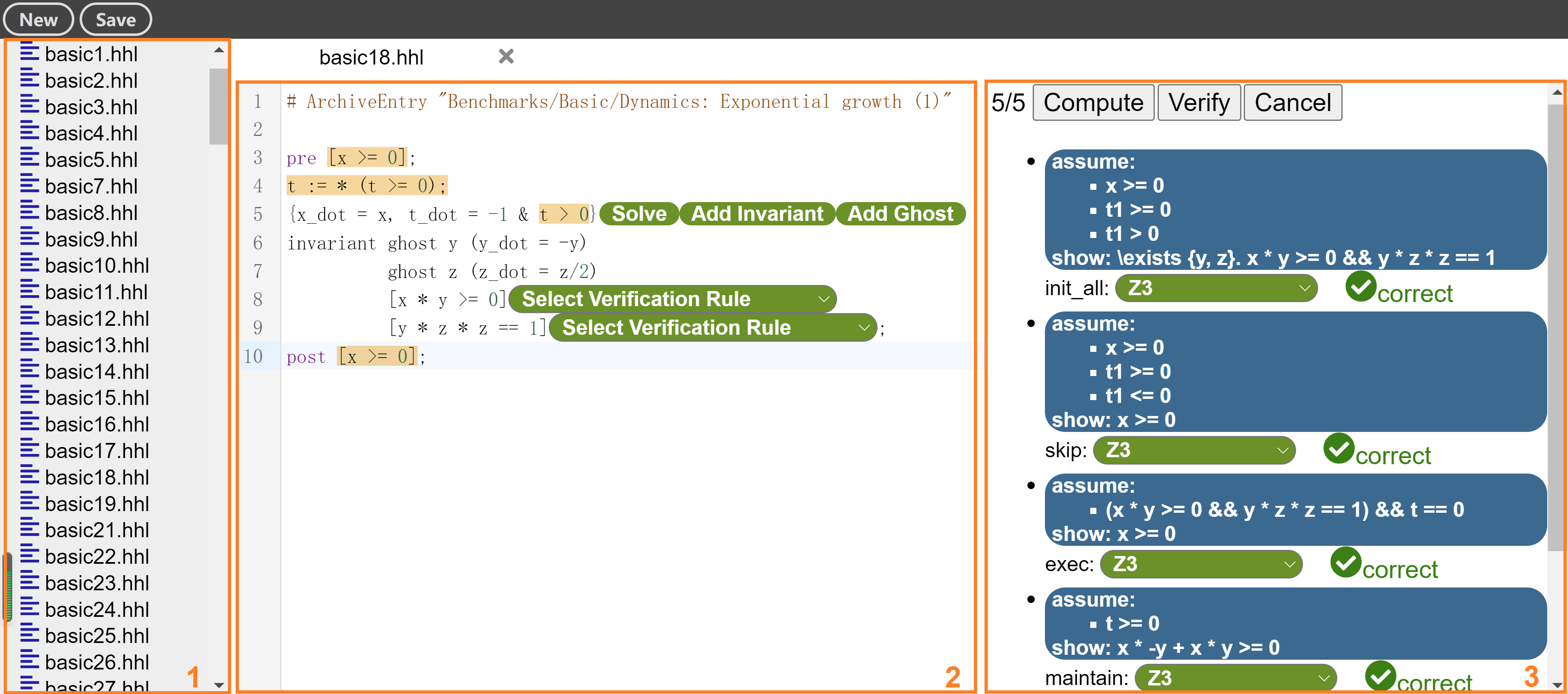}
    \caption{Screenshot of user interface. The left panel (1) shows a list of example files. The middle panel (2) is the editor area, where user can edit the program and add annotations either directly as text or by clicking on buttons. The right panel (3) shows the VCs. When hovering over each VC, the relevant part of the code is highlighted.}
    \label{fig:screenshot}
\end{figure}

\subsection{Evaluation on Simulink/Stateflow Models}

To illustrate the use of our tool as part of an existing toolchain to verify correctness of Simulink/Stateflow models, we show two example models, one from Simulink and one from Stateflow.

\subsubsection{Cruise Control System}

The first example is a cruise control system of an automotive vehicle~\cite{Kekatos-2021-verify}. The system stabilizes the speed of a vehicle around some desired speed~(15 m/s in our case) using a PI controller. The PI controller adjusts the control force according to the difference between actual speed and desired speed as well as its integral. The vehicle follows its physical dynamics. The Simulink models are presented in Appendix~\ref{sec:simulink-models}.

We first applied the approach by Xu et al.~\cite{SimulinkTranslate} to translate the Simulink models into an HCSP program. In the program, the controller and vehicle dynamics are combined into a single ODE. Given the initial speed $v=14$ and initial integral value $I = 700$ of the controller, which are close to the stable point ($v=15, I=750$), we want to verify that the speed remains in the interval $[13.5, 16.5]$.

To verify the Hoare triple, we annotated the ODE and loop with the invariant $1.3*(I-750)^2-198* (I-750)*(v-15)+12192*(v-15)^2 \leq 5542$, and used (dI) rule to prove the ODE invariant. The invariant was derived following the standard theory for analyzing linear dynamical systems. The annotated HCSP program is illustrated below. The tool generated seven VCs, and Z3 can prove all of them.

\begin{lstlisting}[language=HHLPy, columns=flexible]
pre [v == 14][I == 700];

t := 0;
_tick := 0;
tt := 0;

{
  {tt_dot = 1, I_dot = (15 - v) * 40, v_dot = ((15 - v) * 600 + I - v * 50) * 0.001 & tt < 1}
      invariant [1.3*(I-750)^2 - 198 * (I-750)*(v-15)+12192*(v-15)^2<=5542];
  t := t + tt;
  _tick := _tick + 1;
  tt := 0;
}*
  invariant [1.3*(I-750)^2 - 198 * (I-750)*(v-15)+12192*(v-15)^2<=5542];
post [v >= 13.5][v <= 16.5];

\end{lstlisting}

\subsubsection{Sawtooth Wave}

The sawtooth wave is a Stateflow model generating a signal that alternates between increasing from 0 to 1 and decreasing from 1 to 0. It illustrates functionality in Stateflow such as hierarchical states and specifying ODEs in a state. The Stateflow model is presented in Appendix \ref{sec:simulink-models}. The signal $x$ follows the ODE ${\dot{x}=y}$, with $y$ switching between 1 and $-1$ per unit time. We want to verify that every time $y$ switches, $x$ is still between 0 and 1.

We translated the Stateflow model into an HCSP program with the approach by Guo et al.~\cite{StateflowTranslate} (code shown in Appendix \ref{sec:simulink-models}). To verify the program, the loop is annotated with four invariants (mostly having to do with the relationship between Stateflow locations and value of variable $x$), and the ODE is annotated with ``solution''. A total of 62 VCs are generated and proved to be true by Z3.

\subsection{Evaluation on Benchmarks from KeYmaera X}

We also evaluated our tool on the basic and nonlinear benchmarks\footnote{The benchmarks are available at \url{https://github.com/LS-Lab/KeYmaeraX-projects/tree/master/benchmarks}} from KeYmaera X. We first translated the examples from \dL\ to HCSP manually, trying to maintain semantic equivalence as much as possible. Due to the differences between \dL\ and HCSP, some examples can not be translated into HCSP programs. We annotated the programs with invariants and proof rules, mostly following the existing proofs in KeYmaera X.

Given the annotations, \textsf{HHLPy} can verify 50 out of 60 examples in the basic benchmarks. In comparison, KeYmaera X solves 58 examples in the scripted mode (with detailed proof scripts), and 55 examples in the hints mode (with invariants annotated in the model)~\cite{Mitsch-2021-ARCH}. Of the ten unsolved examples, we are unable to translate eight of them to HCSP due to use of \dL-specific constructs; one is non-polynomial; and the last one makes use of invariants containing old versions of variables. For the nonlinear benchmarks, \textsf{HHLPy} can verify 103 out of 141 examples (compared to 108 in the scripted mode and 95 in the hints mode for KeYmaera X~\cite{Mitsch-2021-ARCH}). Most of the unsolved ones are because we are unable to find the invariants or their VCs cannot be proved in reasonable time by Z3 or Wolfram Engine. Specifically, \textsf{HHLPy} can verify 9 examples which KeYmaera X cannot verify in hints mode, while KeYmaera X can verify one example in hints mode that \textsf{HHLPy} cannot verify. For this one example, we have not found the invariants or specific rules, while KeYmaera X verifies it using a general \texttt{ODE} rule. 

In the 153 examples solvable by \textsf{HHLPy}, the user only needs to add annotations including loop/ODE invariants and ODE rules; just a couple of annotations are needed per problem. For some problems, it is necessary to switch the backend solver from the default Z3 to Wolfram Engine. After this, \textsf{HHLPy} can finish the proof automatically. These experiments show that our tool can be used to solve a wide range of examples from existing benchmarks with little manual effort. Moreover, from the evaluation on the benchmarks we note that there are some VCs that Z3 can solve but Wolfram Engine cannot, and vice versa, showing the two solvers have complementary 
advantages. Generally speaking, Z3 handles complex boolean structures better, while Wolfram Engine has advantages in expressions containing many decimal numbers.

\begin{rep}
We compared the performance between HHLPy and KeYmaera X in scripted (with detailed proof scripts) and hints (with invariants annotated in the model) modes\cite{Mitsch-2021-ARCH}. As HHLPy verifies programs annotated with invariants and methods for ODEs, it is in a mode between scripted and hints modes. As shown in Table~\ref{results}, our tool performs competitively with KeYmaera X on the benchmarks from KeYmaera X. 

\begin{table}[!ht]
\caption{Evaluation Results on Basic and Nonlinear Benchmarks. Timeout = 300s. ``Verified" represents the number of examples verified successfully. ``HCSP" represents the number of examples that can be translated into annotated HCSP programs. ``Overall" represents the number of overall examples in benchmarks.}\label{results}
\center
\begin{tabular}{| c | c | c | c | c | c | c |}\hline
\multicolumn{1}{|c|}{\multirow{2}{*}{Tools}} & \multicolumn{3}{c|}{Basic} & \multicolumn{3}{c|}{Nonlinear}\\ \cline{2-7}
 & Verified  & HCSP & Overall & Verified & HCSP & Overall \\ \hline
HHLPy & 50 & 52 & \multirow{3}{*}{60} & 
100 & 141 & \multirow{3}{*}{141}  \\ \cline{1-3}\cline{5-6}
KeYmaera X(scripted) & 58 & \_ & \multirow{3}{*}{} & 108 & \_ & \multirow{3}{*}{} \\
\cline{1-3}\cline{5-6}
KeYmaera X(hints) & 55 & \_ & \multirow{3}{*}{} & 95 & \_ & \multirow{3}{*}{}\\
\hline
\end{tabular}
\end{table}
\end{rep}

\section{Related Work}
\label{sec:related-work}

Differential dynamic logic (\dL)~\cite{Platzer08,Platzer18} models hybrid systems by extending dynamic logic with continuous evolution. Reasoning rules about continuous evolution include differential invariants, differential weakening, differential cut, and differential ghosts. The rules are stated in the form of a uniform substitution calculus~\cite{Platzer17}, and they are complete~\cite{PlatzerT20}. Differential dynamic logic has been implemented in KeYmaera~\cite{PlatzerQ08} and KeYmaera X~\cite{Fulton-2015-KeYmaera}, whose user interfaces display current subgoals in sequent calculus form and allow users to point and click to construct proofs. The Bellerophon language allows users to perform proofs using a tactic language~\cite{FultonMBP17}. The KeYmaera X tool produces proofs that can be independently checked in Isabelle and Coq~\cite{BohrerRVVP17}. Liebrenz et al. developed a method to translate Simulink models to \dL\ and to verify them in KeYmaera X~\cite{LiebrenzHG18}.

Huerta y Munive and Struth represented \dL\ programs using Kleene algebras, and built verification components for hybrid systems in Isabelle/HOL~\cite{MuniveS18,MuniveS22}. Foster et al. proposed Hoare logic rules and refinement calculi for hybrid programs~\cite{FosterMS20}
and extended the verification components in Isabelle/HOL~\cite{FosterMGS21}, e.g., with syntax translation to obtain more user-friendly modeling and specification languages and with proof automation using Eisbach. Huerta y Munive and Struth also described formalization of solutions to affine and linear systems of ODEs, with applications to verifying correctness of such systems~\cite{Munive20}.

Both of the above series of works focused on hybrid programs modeled using \dL. As discussed in Section~\ref{sec:preliminaries} and~\ref{sec:proof-rules}, the semantics of continuous evolution is different from that in \dL; hence the proof rules need to be adapted, resulting in particular to significant changes to differential weakening and differential cut rules. In addition, compared to KeYmaera X in scripted mode, \textsf{HHLPy} can finish the verification automatically once programs are annotated with loop/ODE invariants and ODE rules. Compared to KeYmaera X in hints mode, \textsf{HHLPy} shows the VCs that cannot be proved, and highlights the set of code fragments that contribute to generating the VCs, which help users to debug programs and annotations.

Goncharov and Neves introduced the \textsc{HybCore} language for hybrid computation~\cite{GoncharovN19}. Similarly to HCSP, \textsc{HybCore} defines deterministic semantics for domain constraints of ODE. The connection with Moggi's work on computational effects~\cite{iandc/Moggi91} potentially aids reasoning and verification in \textsc{HybCore}. However, concrete verification methods remain future work.

Compared to the previous version of hybrid Hoare logic~\cite{WangZZ15}, we focus only on the sequential fragment of HCSP, resulting in much simpler rules that permits automatic VC generation. On the other hand, we consider a full set of reasoning rules for ODEs in parallel with \dL, rather than only the invariant rule in~\cite{WangZZ15}.

The design of our tool is similar to many other (semi-)automatic program verification tools, such as Dafny~\cite{Leino10}, VeriFast~\cite{JacobsSPVPP11}, and Why3~\cite{FilliatreP13}, in that annotations are inserted into the program code. Our work differs from these tools firstly in being able to handle hybrid programs. Moreover, we designed detailed labeling and highlighting mechanisms to improve robustness of the annotations and help visualization. These improvements are not limited to hybrid programs, and can potentially be incorporated into other program verification tools as well.

\section{Conclusion}
\label{sec:conclusion}

We presented \textsf{HHLPy}, a tool for verification of hybrid programs written in the sequential fragment of HCSP. The backend of the tool implements a Hoare logic that includes rules for reasoning about continuous evolution adapted from \dL. We also designed labeling and highlighting mechanisms to improve user interaction. We demonstrated the capabilities of the tool on HCSP programs translated from Simulink/Stateflow models and on KeYmaera X benchmarks.

We leave extension of the deduction system to handle communication, interrupts, and parallel composition to future work. On the side of implementation and applications, we intend to further extend the tool to be able to handle non-polynomial ODEs and invariants and permit interactive proofs of VCs.

\paragraph{Acknowledgments.}

This work is supported by the National Natural Science Foundation of China under grant No.\ 62032024, 62002351, and a Chinese Academy of Sciences President’s International Fellowship for Postdoctoral Researchers under grant No.\ 2021PT0015.

%
%
%
\bibliographystyle{splncs04}
\bibliography{bib}

\begin{appendix}
\appendix
\section{More on the Proof Rules}
\label{sec:more-proof-rules}

In this section, we provide more details on the proof rules.

\subsection{Proof Rules for Other Commands}
\label{subsec:more-hoare-rules}

The proof rules for commands other than the ODE are given in Fig.~\ref{fig:hoare-rules}.

\begin{figure}
\[
    \begin{array}{c}
    \prftree[r]{skip}{\hoare{P}{\skipcmd}{P}} \qquad
    \prftree[r]{assign}{\hoare{P[e/x]}{x := e}{P}} \qquad
    \prftree[r]{ichoice}{\hoare{P_1}{c_1}{Q}}{\hoare{P_2}{c_2}{Q}}{\hoare{P_1\wedge P_2}{c_1\ichoice c_2}{Q}}
    \vspace{2mm} \\
    \prftree[r]{nondet-assign}{\hoare{B[y/x]\rightarrow Q[y/x]}{x := *(B)}{Q}} \qquad
    \prftree[r]{seq}{\hoare{P}{c_1}{Q}}{\hoare{Q}{c_2}{R}}{\hoare{P}{c_1;c_2}{R}}
    \vspace{2mm} \\
    \prftree[r]{ite}{\hoare{P_1}{c_1}{Q}}{\hoare{P_2}{c_2}{Q}}{\hoare{(B\rightarrow P_1)\wedge (\neg B\rightarrow P_2)}{\itecmd{B}{c_1}{c_2}}{Q}} \vspace{2mm}\\
    \prftree[r]{loop}{\hoare{I}{c}{I}}{\hoare{I}{\loopcmd{c}}{I}} \qquad
    \prftree[r]{conseq}{P' \rightarrow P}{\hoare{P}{c}{Q}}{Q \rightarrow Q'}{\hoare{P'}{c}{Q'}}
    \end{array}
\]
    \caption{Hoare rules for the non-ODE commands}
    \label{fig:hoare-rules}
\end{figure}

\subsection{Example for the (dW) Rule}

We give some examples of applying the (dW) rule with different invariants, such as the trivial invariants $\mathsf{true}$ and $\mathsf{false}$ or the non-trivial invariant $x-y=1$.

\begin{example}
In this example, the precondition guarantees that the ODE starts inside the domain $D$. The Hoare triple is:
\[ \hoare{x < 5}{\ode{\dot{x}=1}{x<5}}{x=5}\]
As the boundary of domain $D$ implies the postcondition, no invariant is needed (set $I=\mathsf{true}$ by default). From the postcondition $x=5$, we obtain precondition $\neg x<5\rightarrow x=5$ and VC $x=5\rightarrow x=5$. So the VCs are:
\[
\begin{array}{c}
x < 5 \rightarrow \neg x<5 \rightarrow x=5 \\
x=5\rightarrow x=5
\end{array}
\]
\end{example}

\begin{example}
In this example, the precondition guarantees that the ODE starts outside the domain $D$. The Hoare triple is:
\[ \hoare{x>10}{\ode{\dot{x}=1}{x<5}}{x>8} \]
Here we set invariant to be $\mathsf{false}$. From the postcondition $x>8$, we obtain precondition $\neg x<5\wedge x>8$ and no VCs from the ODE. The overall VCs are:
\[
\begin{array}{c}
x>10\rightarrow \neg x<5\wedge x>8
\end{array}
\]
\end{example}

\begin{example}
Now consider a case where the starting position may or may not lie within the domain $D$. The Hoare triple is:
\[ \hoare{x<6}{\ode{\dot{x}=1}{x<5}}{x\ge 5\wedge x<6}\]
No invariant is needed ($I=\mathsf{true}$ is set by default). From the postcondition, we obtain the precondition $\neg x<5\to x\ge 5\wedge x<6$, and the VC $x=5\rightarrow x\ge 5\wedge x<6$. So the overall VCs are:
\[
\begin{array}{c}
x<6\rightarrow \neg x<5\rightarrow x\ge 5\wedge x<6 \\
x=5\rightarrow x\ge 5\wedge x<6
\end{array}
\]
\end{example}

\begin{rep}
\begin{example}
Now consider a case that a non-trivial invariant should be introduced. The Hoare triple is:
\[ \hoare{x^2+y^2=1}{\ode{\dot{x}=-y,\dot{y}=x}{x>0}}{x^2+y^2=1} \]
We set the invariant to be $x^2+y^2-1=0$, requiring to show
\[ \inv{x\ge 0}{\dot{x}=-y,\dot{y}=x}{x^2+y^2-1=0} \] 
And we get the VC $x=0\wedge x^2+y^2-1=0\rightarrow x^2+y^2=1$ and two preconditions $x>0 \rightarrow x^2+y^2-1=0$ and $\neg x>0 \rightarrow x^2+y^2 =1$.
\end{example}
\end{rep}

\begin{example}\label{ex:dw-nont}
Now we consider an example where a non-trivial invariant must be introduced. The Hoare triple is:
\[ \hoare{x=0\wedge y=0}{\ode{\dot{x}=1,\dot{y}=1}{x<5}}{y=5} \]
We set the invariant to be $x-y=0$, requiring to show
\[ \inv{x\le 5}{\dot{x}=1,\dot{y}=1}{x-y=0} \]
We get the VC $x=5\wedge x-y=0\rightarrow y=5$ and the precondition $(x<5 \rightarrow x-y=0) \land (\neg x<5 \rightarrow y=5)$.
\end{example}



\subsection{Example and Proof for the (dI) Rule}

We give an example showing the use of the (dI) rule.
\begin{rep}
\begin{example}[Example~\ref{ex:dw-nont} continued]
In this example, we continue with the invariant triple
\[ \inv{x\ge 0}{\dot{x}=-y,\dot{y}=x}{x^2+y^2-1=0} \]
The Lie derivative of $x^2+y^2-1$ is $2x\dot{x}+2y\dot{y}=-2xy+2xy=0$. So the premise of the (dI) rule is $x\ge 0\rightarrow 0=0$.
\end{example}
\end{rep}

\begin{example}[Example~\ref{ex:dw-nont} continued]
In this example, we continue with the invariant triple we get:

\[ \inv{x\le 5}{\dot{x}=1,\dot{y}=1}{x-y=0} \]

The Lie derivative of $x-y$ is $\dot{x}-\dot{y}=1-1=0$. So the premise of the (dI) rule is $x\le 5\rightarrow 0=0$.

\end{example}

\begin{proof}[of the (dI) rule]
We consider each of the rules in turn.

\begin{description}
\item[Rule dI$_=$] Consider a solution $\boldsymbol{\gamma}:[0,T]\to\mathbb{R}^n$ of the ODE, satisfying the condition that $\boldsymbol{\gamma}(t)$ satisfies $P$ for all $t\in [0,T]$. Then, the derivative of $f(\boldsymbol{\gamma}(t))$ is $0$ for all $t\in[0,T]$. This shows that $f(\boldsymbol{\gamma}(t))$ is constant within this interval. If furthermore $f(\boldsymbol{\gamma}(0))=0$, then we get $f(\boldsymbol{\gamma}(T))=0$, as desired.

\item[Rule dI$_\succcurlyeq$] Similar to before, except now the derivative of
$f(\boldsymbol{\gamma}(t))$ is greater than or equal to $0$ for all $t\in [0,T]$, so
$f(\boldsymbol{\gamma}(t))$ is a non-decreasing function of $t$ within this interval. If
furthermore $f(\boldsymbol{\gamma}(0))\succcurlyeq 0$, then we get
$f(\boldsymbol{\gamma}(T))\succcurlyeq 0$, as desired.

\item[Rule dI$_{\neq}$] Similar to before, we get that the derivative of $f(\boldsymbol{\gamma}(t))$ is $0$ for all $t\in [0,T]$. This shows that $f(\boldsymbol{\gamma}(t))$ is constant within this interval. If furthermore $f(\boldsymbol{\gamma}(0)) \neq 0$, then we get $f(\boldsymbol{\gamma}(T)) \neq 0$, as desired.
\end{description}
\end{proof}

\subsection{Example and Proof for the (dC) Rule}

\begin{example}
In this example, we insert an intermediate invariant to help prove the final invariant. The invariant triple is:
\[ \inv{y\le 1}{\dot{x}=y, \dot{y}=1}{y>0 \wedge x>0} \]
Using the (dC) rule, the two premises are:
\[ \inv{y\le 1}{\dot{x}=y, \dot{y}=1}{y>0}\]
\[ \inv{y \le 1 \wedge y > 0}{\dot{x}=y, \dot{y}=1}{x>0}\]
Both the two premises can be proved by using the (dI) rule.
     
\end{example}

\begin{proof}[of the (dC) rule]
Consider a solution $\boldsymbol{\gamma}:[0,T]\to\mathbb{R}^n$ of the ODE, satisfying the condition that $\boldsymbol{\gamma}(t)$ satisfies $P$ for all $t\in[0,T]$, and $\boldsymbol{\gamma}(0)$ satisfies $Q_1\wedge Q_2$. We wish to show $\boldsymbol{\gamma}(t)$ satisfies $Q_1\wedge Q_2$ for all $t\in [0,T]$. By the first premise, we get $\boldsymbol{\gamma}(t)$ satisfies $Q_1$ for all $t\in [0,T]$. Then by the second premise, $\boldsymbol{\gamma}(t)$ satisfies $Q_1\wedge Q_2$ for all $t\in [0,T]$.
\end{proof}

\subsection{Example and Proof for the (dG) Rule}

\begin{example}
In this example, the Hoare triple is:
\[
\hoare{x > 0\wedge t=0}{\ode{\dot{x}=x, \dot{t}=1}{t<10}}{x > 0}
\]
Using the (dG) rule, we introduce a ghost variable $y$ with the differential equation $\dot{y}=-\frac{y}{2}$ and set the invariant for the new differential equations to be $xy^2=1$. Then the new goal is:
\[
 \inv{t<10}{\dot{x}=x, \dot{y}=-\frac{y}{2}, \dot{t}=1}{xy^2=1}
\]
We get the VC $t=10 \wedge xy^2=1 \rightarrow x>0$ and the precondition $(t<10\rightarrow \exists y.xy^2=1) \land (\neg t<10 \rightarrow x>0)$.

\end{example}

\begin{rep}
\begin{example}
In this example, two ghost variables are added. The Hoare triple is:
\[
\hoare{x\ge 0\wedge t =0}{\ode{\dot{x}=x, \dot{t}=1}{t<10}}{x\ge 0}
\]
We invoke (dG) rule to introduce two ghost variables: $y$ with ODE $\dot{y}=-y$ and $z$ with ODE $\dot{z}=\frac{z}{2}$, with $I$ being $xy \geq 0 \wedge yz^2=1$. Then the new goal is:

\[
\inv{t<10}{\dot{x}=x,\dot{y}=-y, \dot{z}=\frac{z}{2}, \dot{t}=1}{xy\ge 0\wedge yz^2=1}
\]
Then we can use (dC) and (dI) rule to prove the above invariant triple.

The derived precondition is:
\[
(t<10\rightarrow \exists y,z.\ xy\geq 0 \wedge yz^2=1)\wedge(\neg t<10 \rightarrow x\geq0)
\]
The overall VCs are:
\[
\begin{array}{c}
     x \geq 0\wedge t=0 \rightarrow (t<10\rightarrow \exists y,z.\ xy\geq 0 \wedge yz^2=1)\wedge(\neg t<10 \rightarrow x\geq0)\quad \mbox{(precondition)}  \\
     t=10 \wedge xy\geq 0 \wedge yz^2=1 \rightarrow x\geq 0\quad \mbox{(from the (dG) rule)} \\
     t<10\rightarrow xy+x(-y)\geq 0\quad \mbox{(from (dI) and (dC) rule)}\\
     t<10\wedge xy\geq 0 \rightarrow(-y)z^2+2yz\frac{z}{2}=0 \quad \mbox{(from (dI) and (dC) rule)}
\end{array}
\]

The proof script can be written as:
\begin{verbatim}
pre [x >= 0][t == 0]
{x_dot = x, t_dot = 1 & t < 10}
invariant
  ghost y (y_dot = -y)
  ghost z (z_dot = z/2)
    [x * y >= 0]
    [y * z * z == 1]
post [x >= 0]
\end{verbatim}
\end{example}

\end{rep}

\begin{proof}[of the (dG) rule]
Given starting state $\boldsymbol{x}_0$, we distinguish two cases based on whether $\boldsymbol{x_0}$ satisfies $D$.
\begin{enumerate}
    \item If $\boldsymbol{x}_0$ satisfies $D$, then there exists a solution $\boldsymbol{\gamma}: [0,T]\to\mathbb{R}^n$ of the ODE $\langle\dot{\boldsymbol{x}}=\boldsymbol{e}\rangle$ with initial condition $\boldsymbol{x}(0) = \boldsymbol{x}_0$, such that $\boldsymbol{\gamma}(t)$ satisfies $D$ for $t\in [0, T)$ and $\boldsymbol{\gamma}(T)$ satisfies $\neg D$. We wish to show that $\boldsymbol{\gamma}(T)$ satisfies $Q$. Since $D\rightarrow \exists \boldsymbol{y}.I$ holds in the precondition and $\boldsymbol{x}_0$ satisfies $D$, using $\boldsymbol{r}$ to denote a witness of $\boldsymbol{y}$, we get that $(\boldsymbol{x}_0, \boldsymbol{r})$ satisfies $I$. With $\boldsymbol{f}(\boldsymbol{x},\boldsymbol{y})$ satisfying Lipschitz condition, let $\boldsymbol{\psi}: [0,T]\to\mathbb{R}$ denote the solution of ODE $\langle\dot{\boldsymbol{y}}=\boldsymbol{f}(\boldsymbol{x},\boldsymbol{y})\rangle$ with initial value $\boldsymbol{r}$. Given that $\boldsymbol{y}$ are fresh variables and do not limit the duration of the solutions of $\boldsymbol{x}$, we get that $(\boldsymbol{\gamma}, \boldsymbol{\psi})$ is the solution of the ODE $\langle\dot{\boldsymbol{x}}=\boldsymbol{e}, \dot{\boldsymbol{y}}=\boldsymbol{f}(\boldsymbol{x}, \boldsymbol{y})\rangle$ with initial condition $(\boldsymbol{x}_0, \boldsymbol{r})$. Then from $\inv{\overline{D}}{\dot{\boldsymbol{x}}=\boldsymbol{e}, \dot{\boldsymbol{y}}=\boldsymbol{f}(\boldsymbol{x}, \boldsymbol{y})}{I}$, we get that $(\boldsymbol{\gamma}(t), \boldsymbol{\psi}(t))$ satisfies $I$ for $t\in[0, T]$. From $\partial D \wedge I \rightarrow Q$ and $(\boldsymbol{\gamma}(T), \boldsymbol{\psi}(T))$ satisfies $I$ and $\partial D$, we get that $(\boldsymbol{\gamma}(T), \boldsymbol{\psi}(T))$ satisfies $Q$. Since $\boldsymbol{y}$ does not occur in $Q$, we get that $\boldsymbol{\gamma}(T)$ satisfies $Q$, as desired.
    \item If $\boldsymbol{x}_0$ does not satisfy $D$, the proof is similar to the one for the (dW) rule.
\end{enumerate}
\end{proof}

\subsection{Example and Proof for the Barrier Certificate Rule}

\begin{example}
In this example, we use the (bc) rule to prove the following invariant triple:
\[ \inv{t\le 10}{\dot{x}=x^3+x^4}{x^3>5} \]
Then the premise of the (bc) rule is 
\[ t\le 10 \wedge x^3=5 \rightarrow 3x^2(x^3+x^4)>0 \]

\end{example}

\begin{proof}[of the (bc) rule]
The proof is by contradiction. Consider a solution $\boldsymbol{\gamma}:[0, T] \to\mathbb{R}^n$ of the ODE, satisfying the condition that $\boldsymbol{\gamma}(t)$ satisfies $P$ for all $t\in [0,T]$, and $\boldsymbol{\gamma}(0)$ satisfies $f \ge 0$. Assume there exists $t' \in [0,T]$ such that $f(\boldsymbol{\gamma}(t')) < 0$. According to the proof of strict barrier certificate~\cite{PrajnaJP07} in \cite{ZhaoZCLW21}, we get that there exists $t_{sup} \in [0, t')$ satisfying $f(\boldsymbol{\gamma}(t_{sup}))=0$ and the Lie derivative of $f$ at $t_{sup}$ is non-positive. However, according to the premise, $\dot{f} > 0$ when $f=0$, which leads to a contradiction. Therefore, there cannot exist $t' \in [0,T]$ such that $f(\boldsymbol{\gamma}(t')) < 0$, i.e. for all $t \in [0,T]$, $f(\boldsymbol{\gamma}(t')) \ge 0$.

The proof is similar when considering the case where the invariant is $f > 0$.
\end{proof}

\subsection{Example and Proof for the Darboux Rule}

\begin{example}
In this example, we use the (dbx$_=$) rule. The invariant triple is:
\[ \inv{t\le 1}{\dot{x}=5x^2+3x, \dot{z}=5zx+3z, \dot{t}=1}{x+z=0}\]
Then, using the (dbx$_=$) rule with $g = 5x+3$, we obtain the following premise of the (dbx$_=$) rule:
\[ t \le 1 \rightarrow 5x^2+3x+5zx+3z=(5x+3)(x+z)\]
\end{example}

\begin{example}
In this example, we use (dbx$_\succcurlyeq$) rule. The invariant triple is:
\[ \inv{t \le 1}{\dot{x}=-x+1, \dot{t}=1}{x>0}\]
Then, by using the (dbx$_\succcurlyeq$) rule with which $g = -1$, we obtain the following premise of the (dbx$_\succcurlyeq$) rule:
\[ t \le 1 \rightarrow -x+1 \ge -x\]
\end{example}

\begin{proof}[of Darboux rule]
We consider each rule in turn.

\begin{description}
\item[Rule (dbx$_=$)] Consider a solution $\boldsymbol{\gamma}:[0,T]\rightarrow \mathbb{R}^n$ to the ODE, satisfying the condition that $\boldsymbol{\gamma}(t)$ satisfies $P$ for all $t \in [0,T]$, and $\boldsymbol{\gamma}(0)$ satisfies $f=0$. We wish to show that $f(\boldsymbol{\gamma}(t))=0$ for all $t \in [0,T]$. 
According to \cite{lics/PlatzerT18}, given $\dot{f}=gf$, $f=0$ stays invariant along the (analytic) solution to the ODE. Therefore, $\boldsymbol{\gamma}(t)$ satisfies $f=0$ for all $t \in [0,T]$.


\item[Rule (dbx$_\succcurlyeq$)]This rule can be derived by (dI), (dC), (dG) rule according to \cite{lics/PlatzerT18}.
\end{description}
\end{proof}

\subsection{Example and Proof for the Solution Rule}

\begin{example}
We illustrate rule (sln) with the following example:
\[
\hoare{x=0\wedge y=0}{\ode{\dot{x}=1,\dot{y}=x}{x<2}}{y=2}
\]
The solution $\boldsymbol{u}(\tau,x_0,y_0)$ is given by $(\tau+x_0,\frac{1}{2}\tau^2+x_0t+y_0)$. Hence $D(\boldsymbol{u}(\tau,x_0,y_0))$ is $\tau+x_0<2$, and $P'(x)$ is
\[ \forall t>0.\, (\forall 0\le\tau<t.\, \tau+x<2) \wedge \neg (t+x<2) \rightarrow \frac{1}{2}t^2+xt+y=2 \]
The two premises $\forall 0\le\tau<t.\, \tau+x<2$ and $\neg (t+x<2)$ of the implication are together equivalent to $t+x=2$. Hence, under the assumption of $x<2$ (so that $t>0\wedge t+x=2$ has a unique solution), the predicate $P'$ reduces to
\[ \frac{1}{2}(2-x)^2 + x(2-x) + y = 2 \]
So the precondition we derive from the solution rule is $(x<2\rightarrow \frac{1}{2}(2-x)^2+x(2-x)+y=2) \wedge (\neg x<2\rightarrow y=2)$. Since $x=0\wedge y=0$ implies this condition, we prove the Hoare triple.

\end{example}

\begin{proof}[of solution rule]
Given starting state $\boldsymbol{x}_0$, we distinguish two cases based on whether $\boldsymbol{x}_0$ satisfies domain $D$.
\begin{enumerate}
\item If $\boldsymbol{x}_0$ satisfies $D$, then there exists a solution $\boldsymbol{\gamma}:[0,T]\to\mathbb{R}^n$ of the ODE $\dot{\boldsymbol{x}}=e$ with initial state $\boldsymbol{x}_0$, such that $\boldsymbol{\gamma}(\tau)$ satisfies $D$ for $\tau\in[0,T)$ and $\boldsymbol{\gamma}(T)$ satisfies $\neg D$, and we wish to show that $\boldsymbol{\gamma}(T)$ satisfies $Q$. Since $D\rightarrow P'$ holds in the precondition, we get that $\boldsymbol{x}_0$ satisfies $P'$:
\[ \forall t>0.\, (\forall 0\le\tau<t.\, D(\boldsymbol{u}(\tau,\boldsymbol{x}_0)))\wedge \neg D(\boldsymbol{u}(t,\boldsymbol{x}_0)) \rightarrow Q(\boldsymbol{u}(t,\boldsymbol{x}_0
)).\]
Then let t = T. We get:
\[\forall 0\le\tau<T.\, D(\boldsymbol{u}(\tau,\boldsymbol{x}_0))\wedge \neg D(\boldsymbol{u}(T,\boldsymbol{x}_0)) \rightarrow Q(\boldsymbol{u}(T,\boldsymbol{x}_0
)).\]
Next, since $\boldsymbol{u}$ is the unique solution to the differential equation $\dot{x}=e$, we get that $\boldsymbol{\gamma}(\tau)=\boldsymbol{u}(\tau,\boldsymbol{x_0})$ for all $0\le \tau\le T$. 
Then we get:
\[\forall 0\le\tau<T.\, D(\boldsymbol{\gamma}(\tau))\wedge \neg D(\boldsymbol{\gamma}(T)) \rightarrow Q(\boldsymbol{\gamma}(T)).\]
From the fact that $\boldsymbol{\gamma}(\tau)$ satisfies $D$ for $\tau\in[0,T)$ and $\boldsymbol{\gamma}(T)$ satisfies $\neg D$, we get $\boldsymbol{\gamma}(T)$ satisfies $Q$, as desired.

\item If $\boldsymbol{x}$ does not satisfy $D$, then the ODE is not executed, and we wish to show that $\boldsymbol{x}$ satisfies $Q$. Since $\neg D\rightarrow Q$ holds in the precondition, $\boldsymbol{x}$ satisfies $Q$, as desired.
\end{enumerate}
\end{proof}

\section{Proof of Soundness of the Verification Condition Generation}
\label{sec:proof-vcg}

\begin{thm-vcg-sound}
\thmvcgsound
\end{thm-vcg-sound}
\begin{proof}
We proceed by structural induction on $\mathcal{T}$.
We assume an initial state in which 
$P_1, \dots, P_m$ hold.
We must show that after executing $\mathcal{T}$ on this state,
if $\mathcal{T}$ terminates,
$Q_1,\dots,Q_n$ hold.

The predicates $\tilde{P}_1,\ldots,\tilde{P}_{\tilde{m}}$ in \vctag{vc}
are the subset of the
preconditions $P_1,\ldots,\allowbreak P_m$ whose variables are never reassigned in $\mathcal{T}$.
They hold in the initial state by assumption and
since their variables never change, they hold in all states in the course of the program $\mathcal{T}$.
Thus, for the purposes of this proof, we can
assume that $\tilde{P}_1,\ldots,\tilde{P}_{m'}$ always hold.
Therefore, we can also assume that the elements of $\vc(\mathcal{T}, \{Q_1,\dots,Q_n\})$ 
always hold.

If $T$ is any command except loop or ODE,
the correctness of the Hoare triple is easy to see
using the induction hypotheses and the VCs emerging via
\vctag{pre}
from
\vctag{pre-skip}, \vctag{pre-assn}, \vctag{pre-seq}, \vctag{pre-if},
\vctag{pre-else}, \vctag{pre-choice}, and \vctag{pre-nassn}.

If $\mathcal{T}$ is a loop $\loopcmd{\mathcal{S}}$ 
with invariants $I_1, \dots, I_n$,
we proceed as follows.
Since $\pre(\mathcal{T}, Q_1,\dots,Q_n) = \{I_1, \dots I_n\}$ by \vctag{pre-loop}, 
we have
$P_1 \land \dots \land P_m \rightarrow I_i 
\in \mathrm{VC}(\hoare{P_1, \dots, P_m}{\mathcal{T}}{Q_1,\dots,Q_n})$ \vctag{pre} for each $I_i$.
Thus, all $I_i$ hold in the initial state.
The induction hypothesis for $\mathcal{S}$,
in conjunction with
\vctag{vc-loop-body} and 
\vctag{vc-loop-maintain},
imply $\hoare{I_1, \dots, I_n}{\mathcal{S}}{I_1,\dots,I_n}$.
Hence, the $I_i$ not only hold in the initial state, but
still after executing $\mathcal{S}$ arbitrarily often.
Finally, by \vctag{vc-loop-exit}, $Q_1,\dots,Q_n$ hold in the final state.

If $\mathcal{T}$ is a an ODE $\ode{\dot{\boldsymbol{x}}=\boldsymbol{e}}{D}$
with invariants $I_1, \dots, I_k$, we proceed as follows.

If the ODE is annotated to use the solution rule, 
we apply it using the VC stemming from the precondition \vctag{pre-sln}, and we are done.

Otherwise, invariants have been specified, or the default invariant $I_1 = \mathsf{true}$ has been set. If additionally ghost variables are specified,
we employ the (dG) rule; if no ghost variables are specified we employ the (dW) rule.
The VCs stemming from \vctag{pre-dWG-skip}
and \vctag{pre-dW-init} or \vctag{pre-dG-init}
show that the rule (dW) or (dG) is applicable to our
desired Hoare triple $\hoare{P_1, \dots, P_m}{\mathcal{T}}{Q_1,\dots,Q_n}$.
The condition
\vctag{vc-dWG-exec} 
discharges the right premise of the (dW) or (dG) rule.
It remains the left premise
\[\inv{\overline{D}}{\dot{\boldsymbol{x}}=\boldsymbol{e},
\dot{\boldsymbol{y}}=\boldsymbol{f}(\boldsymbol{x},\boldsymbol{y})}{I_1\land\cdots \land I_k}\]
(possibly without $\boldsymbol{y}$ in the case of the (dW) rule)

Next, we repeatedly apply the (dC) rule to isolate each invariant $I_i$.
Then it remains to show
\[\inv{\overline{D} \land I_1 \land \dots \land I_{i-1}}{\dot{\boldsymbol{x}}=\boldsymbol{e},
\dot{\boldsymbol{y}}=\boldsymbol{f}(\boldsymbol{x},\boldsymbol{y})}{I_i}\]

We distinguish the different rules that could be in the annotation:
\begin{itemize}
    \item If the user did not specify a rule and $I_i$ is either $\mathsf{true}$
    or $\mathsf{false}$, what we need to show is obvious because 
    $\mathsf{true}$ and $\mathsf{false}$ are always invariants.
    \item dI: We apply the (dI) rule. Its premise is justified by \vctag{vc-dI1}, \vctag{vc-dI2}, or \vctag{vc-dI3}.
    \item dbx: We apply the (dbx) rule. Its premise is justified by \vctag{vc-dbx1} or \vctag{vc-dbx2}.
    \item bc: We apply the (bc) rule. Its premise is justified by \vctag{vc-bc}.
\qed
\end{itemize}
\end{proof}

\section{Simulink/Stateflow Models}
\label{sec:simulink-models}

\subsection{Model for the Cruise Control System Example}
Simulink models for the cruise control system are shown in Fig.~\ref{fig:top}, \ref{fig:car}, and \ref{fig:control}.
\begin{figure}[htbp]
\includegraphics[width=\textwidth]{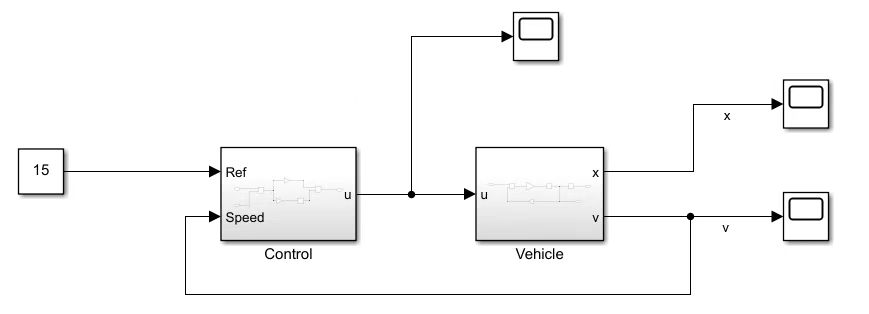}
\caption{Simulink Model of Cruise Control System.} \label{fig:top}
\end{figure}

\begin{figure}[htbp]
\includegraphics[width=\textwidth]{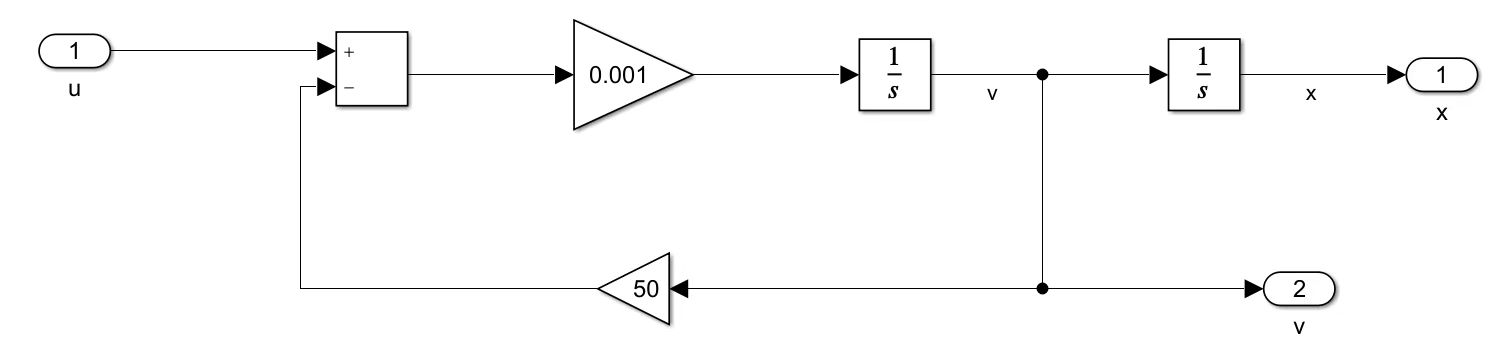}
\caption{Simulink Model of Vehicle Dynamics.} \label{fig:car}
\end{figure}

\begin{figure}[htbp]
\includegraphics[width=\textwidth]{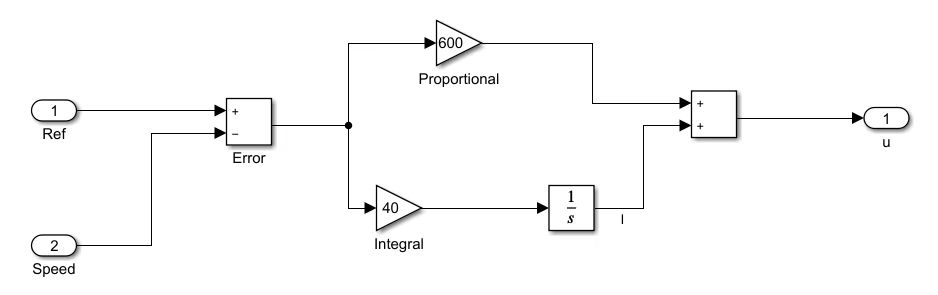}
\caption{Simulink Model of PI Controller.} \label{fig:control}
\end{figure}

\subsection{Model and Code for the Sawtooth Wave Example}
The Stateflow model of the sawtooth wave example is shown in Fig.~\ref{fig:sawtooth}.

\begin{figure}[htbp]
\centering
\includegraphics[scale=0.5]{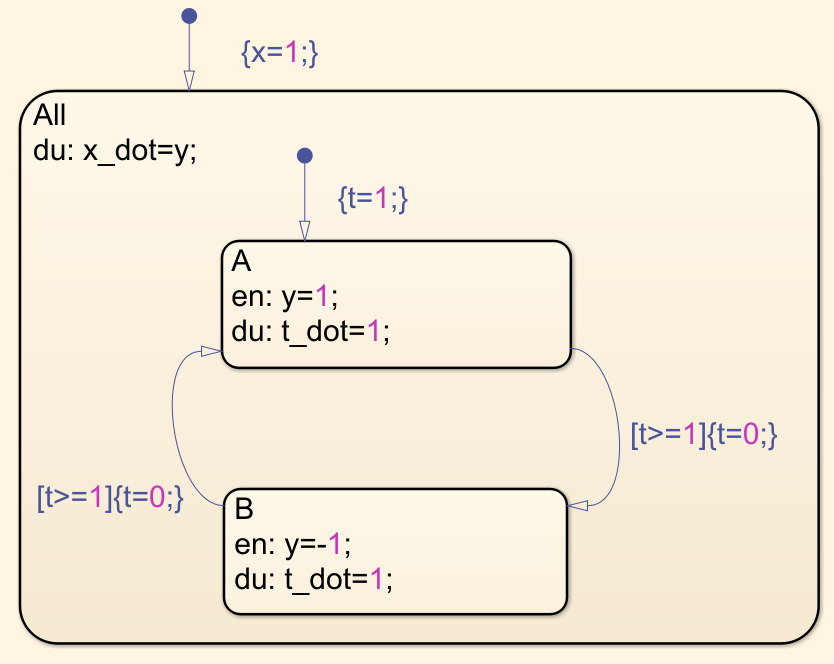}
\caption{Stateflow Model of Sawtooth Wave.} \label{fig:sawtooth}
\end{figure}

Fig.~\ref{fig:sawtooth-code} shows the HCSP code that is automatically translated from the Stateflow diagram for the sawtooth example, together with annotations of pre- and postconditions, invariants, and ODE proof rules used to verify the property that $x$ remains in the interval $[0,1]$ for this model.

\begin{figure}[htbp]
\centering
\begin{lstlisting}[language=HHLPy, columns=flexible]
pre [true];

  Chart_All_A := 1;
  Chart_All_B := 2;
  Chart_All := 0;
  x := 1;
  Chart_st := Chart_All;
  t := 1;
  Chart_All_st := Chart_All_A;
  y := 1;

  {
    Chart_ret := 0;
    if (Chart_All_st == Chart_All_A) {
      Chart_All_A_done := 0;
      if (t >= 1) {
        t := 0;
        if (Chart_All_st == Chart_All_A) {
          Chart_All_st := Chart_All_B;
          y := -1;
          Chart_All_A_done := 1;
        }
      }
      Chart_ret := Chart_All_A_done;
    } else if (Chart_All_st == Chart_All_B) {
      Chart_All_B_done := 0;
      if (t >= 1) {
        t := 0;
        if (Chart_All_st == Chart_All_B) {
          Chart_All_st := Chart_All_A;
          y := 1;
          Chart_All_B_done := 1;
        }
      }
      Chart_ret := Chart_All_B_done;
    }
      {x_dot = y, t_dot = 1 & t < 1} solution;   
  }*
  invariant [Chart_All_st == Chart_All_A -> x == 1]
            [Chart_All_st == Chart_All_B -> x == 0]
            [Chart_All_st == Chart_All_A || Chart_All_st == Chart_All_B]
            [t == 1];

post [x >= 0][x <= 1];


\end{lstlisting}
\caption{HCSP code translated from the Stateflow diagram for the sawtooth example, with annotations.}
\label{fig:sawtooth-code}
\end{figure}
\end{appendix}

\end{document}